\documentclass[journal]{IEEEtran}
\usepackage{color}
\input{mysymbol.sty}

\IEEEoverridecommandlockouts                              

\ifCLASSINFOpdf
\else
\fi

\newcommand{\mc}[1]{\mathcal{#1}}
\usepackage{bm}
\newcommand{\nn}[0]{\nonumber}

\newcommand\norm[1]{\left\lVert#1\right\rVert}

\newcommand{\BEAS}{\begin{eqnarray*}}
\newcommand{\EEAS}{\end{eqnarray*}}
\newcommand{\BEQ}{\begin{equation}}
\newcommand{\EEQ}{\end{equation}}
\newcommand{\BIT}{\begin{itemize}}
\newcommand{\EIT}{\end{itemize}}

\usepackage{amsthm}
\newtheorem{theorem}{Theorem}
\newtheorem{corollary}{Corollary}[theorem]
\newtheorem{lemma}{Lemma} 
\newtheorem{definition}{Definition}[section]
\newtheorem{proposition}{Proposition}

\newtheorem{RK}{Remark}[section]

\newtheorem{assumption}{Assumption}[section]

\usepackage{psfrag}

\usepackage{nomencl}

\makenomenclature

\usepackage{makeidx}
\makeindex

\usepackage{tikz}
\usetikzlibrary{shapes,backgrounds,calc}

\makeatletter
\tikzset{circle split part fill/.style  args={#1,#2}{%
 alias=tmp@name, 
  postaction={%
    insert path={
     \pgfextra{%
     \pgfpointdiff{\pgfpointanchor{\pgf@node@name}{center}}%
                  {\pgfpointanchor{\pgf@node@name}{east}}%
     \pgfmathsetmacro\insiderad{\pgf@x}
      \fill[#1] (\pgf@node@name.base) ([xshift=-\pgflinewidth]\pgf@node@name.east) arc
                          (0:180:\insiderad-\pgflinewidth)--cycle;
      \fill[#2] (\pgf@node@name.base) ([xshift=\pgflinewidth]\pgf@node@name.west)  arc
                           (180:360:\insiderad-\pgflinewidth)--cycle;            
         }}}}}
 \makeatother

\usepackage{pgffor}

\usepackage{array}
\newcolumntype{"}{@{\hskip\tabcolsep\vrule width 1pt\hskip\tabcolsep}}
\newcommand{\thickhline}{%
    \noalign {\ifnum 0=`}\fi \hrule height 1pt
    \futurelet \reserved@a \@xhline
}

\usepackage{caption}
\definecolor{darkgreen}{RGB}{0,153,0}
\definecolor{purple}{RGB}{255,51,255}

\usepackage{algorithmic}
\usepackage{algorithm}
\usepackage{graphicx}
\usepackage{arydshln}
\usepackage{amsmath}
\usepackage{amsfonts}
\usepackage{mathtools}
\usepackage{multirow}
\usepackage{amssymb}
\usepackage{comment}

\graphicspath{{figures/}}
\ifCLASSOPTIONcompsoc
    \usepackage[caption=false, font=normalsize, labelfont=sf, textfont=sf]{subfig}
\else
\usepackage[caption=false, font=footnotesize]{subfig}
\fi

 \renewcommand{\baselinestretch}{0.97}\selectfont

\begin{document}
\title{A Market Mechanism for Trading Flexibility Between Interconnected Electricity Markets}
\author{Hossein Khazaei,
        Ceyhun Eksin, Roohallah Khatami and Alfredo Garcia
}
\maketitle
\thispagestyle{plain}

\begin{abstract}
Electricity markets differ in their ability to meet power imbalances in short notice in a controlled fashion. Relatively flexible markets have the ability to ramp up (or down) power flows across interties without compromising their ability to reliably meet internal demand. In this paper, a market mechanism to enable flexibility trading amongst market operators is introduced. In the proposed market mechanism, market operators exchange information regarding optimal terms of trade (nodal prices and flows) along interconnection lines at every trading round. Equipped with this information, each market operator then independently solves its own internal chance-constrained economic dispatch problem and broadcasts the updated optimal terms of trade for flows across markets. We show the proposed decentralized market mechanism for flexibility trading converges to a Nash equilibrium of the intraday market coupling game, i.e. a combination of internal market clearing solutions (one for each participating market) and flows and prices along interconnection lines so that no individual market operator has an incentive to modify its own internal solution and/or the terms of trade along interties. For a specific class of chance constraints, we show that the limiting equilibrium outcome is efficient, i.e. it corresponds to the solution of the single market clearing problem for all participating markets. The proposed market mechanism is illustrated with an application to the three-area IEEE Reliability Test System.

\end{abstract}
%
\begin{IEEEkeywords}
Market coupling, decentralized optimization, flexibility trading, chance-constrained optimization.
\end{IEEEkeywords}
\IEEEpeerreviewmaketitle
\vspace{-6pt}
\section{Introduction}

\IEEEPARstart{T}{he} interconnection (or coupling) of independently operated electricity markets provides increased efficiency by enabling access to cheaper generation resources.
Market coupling may also bring about increased reliability by potentially equipping individual market operators with the {\em option} to call upon adjustments to the scheduled power exchanges across interconnection lines (or tie-lines) in short notice. 
We refer to the procurement of such options as {\em flexibility} trading, which aligns with the definition of power system flexibility in the literature \cite{epri2016electric}, i.e., the ability of power systems to adapt to dynamic and changing conditions.
In this paper we propose a mechanism for trading flexibility between independently operated markets. 

With the increasing penetration of intermittent renewable capacity in electricity markets, the ability to adjust power consumption or generation in a controlled fashion plays a pivotal role in ensuring reliable market operation. 
Accordingly, intraday market designs have been implemented to address large imbalances between day-ahead (DA) and real-time (RT) markets \cite{BookIntegRenewElecMark}.
The goal of flexibility trading across interconnection lines is to give market operators an {\em additional} tool to address such imbalances.

The market architecture for intraday flexibility trading proposed in this paper is {\em decentralized} i.e. independent market operators retain control over the terms and execution of trades and are able to enforce market-specific standards for reliable operation.
The proposed market architecture is also {\em iterative}, i.e. 
after DA dispatch is made available from the pertinent market, at each iteration of the intraday market, each market operator solves a \emph{chance-constrained} model for the optimal procurement of reserves given the prices for flexibility quoted {\em internally} (e.g. by generators) and {\em externally} by adjacent areas. 

\vspace{-6pt}
\subsection{Literature Review} \label{sec::LitReview}

The European Price Coupling of Regions (PCR) project is perhaps the largest example of a market coupling mechanism currently in operation. In such mechanism, on a day-ahead basis market participants submit their orders to their respective power exchange (PX). These orders are collected and submitted to a centralized market-clearing algorithm ({\em Euphemia}) which identifies the prices and the trades that should be executed so as to maximize the gains from trade generated by the executed orders while ensuring the available transfer capacity is not exceeded \cite{ENTSO-E}. The market mechanism proposed in this paper fundamentally differs from the PCR design in that it is {\em decentralized} and accounts for resource uncertainty.

In the United States, a different approach referred to as coordinated transaction scheduling (CTS) has been approved by FERC and implemented in NYISO, PJM, and MISO. In this method, proxy buses at the interties are used as trading points for enabling bids from external market participants. Thus offers to sell (buy) from external market participants are represented by injections (withdrawals) at such proxy buses \cite{Guo}. However, 
CTS is limited to scheduling interchanges between two neighboring areas. This clearly limits efficiency gains as areas may need to interact with more than one neighboring area simultaneously in order to identify maximum gains from trade.
Unlike CTS, the market mechanism proposed in this paper allows any finite number of market operators to trade flexibility across interties.

A large body of research works have been dedicated to evaluating power systems flexibility, e.g., through
probabilistic metrics in  \cite{lannoye2012evaluation,lannoye2014transmission} and statistical analysis in \cite{huber2014integration}, which define the system flexibility respectively based on insufficient ramping resource expectation and the frequency of ramping events of net-demand (demand minus renewable generation).
Aiming at fusing and unifying different metrics, a comprehensive flexibility metric is also presented in \cite{zhao2015unified}.
Another line of research efforts focus on scheduling flexibility resources by deploying different optimization tools, e.g., stochastic \cite{wu2014thermal,khatami2019flexibility} and chance-constrained \cite{zhang2017chance} optimization.
However, to the best of our knowledge, a decentralized market mechanism for trading flexibility between interconnected electricity markets has not yet been developed.

Considering each market as an independent entity with an objective and local constraints that are coupled with other markets due to the existence of interconnection lines leads to viewing the market coupling problem as a non-cooperative game with coupled constraints. Thus, the optimal solution from the viewpoint of a market is a Nash equilibrium strategy profile that satisfies the local constraints. In this sense, our proposed mechanism is relevant to Nash equilibrium seeking algorithms for non-cooperative games with coupled constraints \cite{fischer2014generalized}. Distributed Nash equilibrium seeking algorithms are developed for convex \cite{pavel2019distributed}, network  \cite{yi2017distributed}, and aggregative games with affined coupled constraints \cite{parise2015network,paccagnan2016distributed} using best-response, or primal-dual type updates.  

The proposed mechanism in this paper is also related to the literature on decentralized optimization algorithms used for solving economic dispatch problems \cite{Binetti14,Liu18}---see \cite{Kargarian18} for a review on the topic. Existing techniques for solving decentralized multi-area economic dispatch problems can be classified into primal decomposition methods and dual decomposition methods.
A ``marginal equivalent" primal decomposition algorithm is proposed in \cite{Litvinov14} for solving the multi-area optimal power flow problem where the information of the marginal units and the binding constraints are exchanged among the local system operators. A generalized Bender decomposition method is used in \cite{LiWu16} in order to decentralize the economic dispatch problem. Finally, in \cite{TongGuo17} the authors propose a coordinated economic dispatch model based on primal decomposition approach where the boundary bus voltage phase angles are the coupling variables.  
In dual decomposition methods, the original problem is converted into separable sub-problems and coupling constraints are dualized (see, e.g.  \cite{Bakirtzis03}, \cite{conejo}). In \cite{erseghe} the alternating direction multiplier method (ADMM) is used to derive a fully distributed and robust algorithm for solving the optimal power flow. In \cite{LuLiu18} distributed interior point method is used to build a fully decentralized optimal power flow algorithm, and by using the unidirectional ring communication graph for information exchange, the upward and downward communication among areas are eliminated.


\subsection{Scope and Contributions}
In this paper our goal is not to propose a new method for Nash equilibrium seeking or for solving the multi-area economic dispatch problem in a decentralized fashion. Rather, we propose a decentralized market mechanism, i.e. a set of rules and protocols governing the physical and economic interaction between market operators in order to enable trading of flexibility services.
The market mechanism proposed in this paper fundamentally differs from the current market coupling design implemented in Europe (PCR) design in that it is {\em decentralized}. In addition, despite the coupling mechanism used in US (CTS) which is limited to two markets, the mechanism introduced in this paper allows {\em any finite number} of market operators to trade flexibility across interties.
Further, the present work is different from our earlier work in \cite{garcia2021incentive} as it accounts for the net-demand uncertainty and enables trading for flexibility, rather than presuming a deterministic net-demand and trading solely for the energy. 

In the proposed mechanism, individual market operators 
independently revisit their preferred terms of trade for flexibility (price and quantity from adjacent markets) at each iteration. 
Individual market preferences over cost {\em vs.} reliability are modeled by means of a chance-constrained formulation of individual market clearing. 
That is, each individual market operator solves a {\em local} chance-constrained economic dispatch problem at each iteration to identify the desired terms of trade with neighboring markets.
The impact of capacity and ramping limitations of local generation resources are reflected on the terms of trade through internalizing the pertinent constraints within individual market clearing processes. In Theorem 1, we show that the proposed market mechanism enables the decentralized identification of a Nash equilibrium of the intraday market coupling game, i.e. a combination of internal market clearing solutions (one for each participating market) and flows and prices along interconnection lines so that no individual market operator has an incentive to modify its own internal solution and/or the terms of trade along interties. We show that Theorem 1 holds on a more general setting with any market clearing model that embeds a convex optimization problem. In Theorem 2, we show that when market operators use a chance-constrained formulation for intraday market clearing, trading outcomes between areas converge to the optimal centralized market clearing solution for all participating markets.


\label{1_Intro}



\section{A Chance-constrained Market Mechanism for Trading Flexibility in Intraday Markets}
\label{2_PropDecentModel}
%
%
%
%

Market operators run different markets before delivering electricity to the customers. These markets include DA market, intraday market, and balancing or RT market. A market operator uses each market to achieve certain set of goals. The purpose of DA market is to schedule and commit slow-ramping energy resources to meet expected demand, while the intraday market is to give the ISO enough flexibility to adjust scheduled dispatch with respect to changes in forecasted demand and/or supply. Finally, the balancing or real-time market is used to ensure that demand and supply are in balance at the time of delivery. In this section, we propose a coupling market mechanism that extends the intraday markets to include flexibility trading with neighboring power grids.
In this section, the proposed market clearing model is developed for a special case where each regional market operator uses a chance-constrained formulation with possibly different reliability requirements. 
We generalize the results of proposed approach for modeling the market clearing by independent market operators as the solution of a {\em black-box} convex optimization problem (Section \ref{sec:BlackBox}).


\subsection{Chance-constrained Market Clearing} 

Let $\mathcal{A}=\{1,\dots,A\}$ denote the index set of different regions and the associated markets. 
Each area $a\in\mathcal A$ is represented by a directed graph $(\mathcal{N}_a,\mathcal{H}_a)$ where $\mathcal{N}_a=\{1,2,\ldots,N_a\}$ and $\mathcal{H}_a=\{(i,j)|i,j \in \mathcal{N}_a,j\equiv j(i)\}$ respectively represent the set of nodes (buses) and edges (transmission lines).
Also, let $\mc{G}_a=\{1,2,\ldots,G_a\}$ denote the set of generators in region $a$, and $\mc{T}_a=\{(i,j)|i\in \mathcal{N}_a, j\in \mathcal{N}_{a'}, a'\in \mathcal A\backslash a\}$ denote the set of tie-lines connected to region $a$.
Intraday clearing for market $a \in \mathcal{A}$ aims to identify adjustments to power generation of all generators $g\in\mathcal{G}_a$ and interconnection flows of all tie-lines $(i,j) \in \mathcal{T}_a$ so as to minimize a given objective $V_a$ while meeting reliability constraints.

We assume each regional market operator $a$ uses the following market clearing model:
\begin{align}
     &\min_{\bm{\Delta} \mathbf{P}_a,\bm{\Delta} \mathbf{T}_a, \bm{\theta}_{a}} V_a = \sum_{g \in \mc{G}_a} C_{a,g} (P^\mathrm{DA}_{a,g} + \Delta P_{a,g})\nn
    \\  &+  \sum_{(i,j) \in \mc{T}_a \cap \mc{T}_{a'}}\hspace{-10pt} \big\{ 
    - \delta_{a',(i,j)}  \Delta T_{a,(i,j)}  + \frac{\mu_{(i,j)}}{2}  |\Delta T_{a,(i,j)}| \big\}, \label{eq::decen_obj_fun} \\
    & \mbox{s.t.}    \nn \\
     &\mathbf{M}_{a} (\mathbf{P}^\mathrm{DA}_{a} + \bm{\Delta} \mathbf{P}_{a})-\mathbf{B}_{a} \bm{\theta}_{a} - \mathbf{R}_{a} (\mathbf{T}^\mathrm{DA}_{a} \!+\! \mathbf{\Delta T}_a) \succeq \mathbf{D}_{a}, \nn\\
     &\hspace{200pt}   
      (\bm{\alpha}_{a}), \label{eq::decen_load_bal_0} \\
     &\underline{\mathbf{P}}_a \leq \mathbf{P}^\mathrm{DA}_a + \bm{\Delta} \mathbf{P}_a \leq \overline{\mathbf{P}}_a, \hspace{75pt} (\bm\nu_{a} , \bm\lambda_{a}), \label{eq::decen_Gen_Lim}\\
    &\underline{\mathbf{R}}_a \leq  \bm{\Delta} \mathbf{P}_a \leq \overline{\mathbf{R}}_a, \hspace{103pt} (\bm\psi_{a} , \bm\varphi_{a}), \label{eq::decen_Ramp_Lim}\\
     &- \overline{H}_{a,(i,j)} \leq \frac{\theta_{a,i} - \theta_{a,j}}{x_{a,(i,j)}} \leq \overline{H}_{a,(i,j)}, \hspace{2pt} (\kappa_{a,(i,j)}, \eta_{a,(i,j)}), \nn \\
    & \hspace{170pt} \forall (i,j) \in \mc{H}_a,\label{eq::decen_L_flow}\\
     &\Delta T_{a,(i,j)} \!= \frac{\theta_{a,i} \!-\! \theta_{a',j}}{\bar{x}_{a,(i,j)}} \!-\! T^\mathrm{DA}_{a,(i,j)},  (\xi_{a,(i,j)}),  
     \forall (i,j) \in \mc{T}_a, \label{eq::decen_T_flow} \\
    &\Pr\big(
    \mathbf{1}^{\top}\! \mathbf{D}_a  \leq 
    \mathbf{1}^{\top}\! (\mathbf{P}^\mathrm{DA}_a \!+\! \bm{\Delta} \mathbf{P}_a) \!-\! \mathbf{1}^{\top}\! (\mathbf{T}^\mathrm{DA}_a \hspace{-1pt} \!+\! \hspace{-1pt} \bm{\Delta} \mathbf{T}_a) \big) \!\geq\!
    1 \!-\! t_a, \nn \\
    & \hspace{200pt} (\gamma_a), \label{eq::decen_load_bal_5} \\
     &\theta_{1,1}=0, \label{eq::decen_theta_1}
\end{align}
where 
 $\mathbf{P}^\mathrm{DA}_{a}=[P^\mathrm{DA}_{a,g}|_{g\in\mathcal{G}_a}]$ is the given DA scheduled power output of generators in $\mathcal{G}_a$ and $\mathbf{\Delta P}_{a}=[\Delta P_{a,g}|_{g\in\mathcal{G}_a}]$ represents the adjustment to power in the intraday market,
 $\mathbf{T}^\mathrm{DA}_{a}=[T^\mathrm{DA}_{a,(i,j)}|_{(i,j)\in \mathcal{T}_a}]$ and $\mathbf{\Delta T}_{a}=[\Delta T_{a,(i,j)}|_{(i,j)\in \mathcal{T}_a}] $ are respectively the scheduled DA and adjusted intraday power flows of tie-lines in $\mathcal{T}_a$,
the vector $\bm\theta_a$ collects the bus voltage phase angles in $\mathcal{N}_a$, the admittance matrix of area $a$ is represented by $\mathbf{B}_a$, and
the incidence matrices $\mathbf{M}_a$ and $\mathbf{R}_a$ respectively map the generators in $\mc{G}_a$ and the tie-line power flows in $\mathcal{T}_a$ to the buses in $\mathcal{N}_a$. 
The vector $\mathbf{D}_a$ denotes the {\em random} forecast for nodal net-demands in region $a$. The vectors $\underline{\mathbf{P}}_a$ and  $\overline{\mathbf{P}}_a$ respectively represent the lower bound and upper bound power limits of the generators in $\mc{G}_a$, and the vectors
$\underline{\mathbf{R}}_a$ and  $\overline{\mathbf{R}}_a$ respectively represent the associated down and up ramp rate limits.
$\overline{H}_{a,(i,j)}$ and $x_{a,(i,j)}$ are respectively the capacity and reactance of transmission lines in $\mathcal{H}_a$, and $\bar{x}_{a,(i,j)}$ is the tie-line reactance for $(i,j) \in \mathcal{T}_a$. 
The parenthesis after each constraint embed the corresponding dual variables.


The objective function in \eqref{eq::decen_obj_fun} includes the total generation cost of generators in region $a$ and the revenue/cost of exporting/importing power to/from the neighboring areas where $\mu_{(i,j)}$ is a capacity price associated to the capacity of the tie-line $(i,j)\in\mathcal{T}_a$, and the price $\delta_{a',(i,j)}$ is the willingness to pay by market $a'$ for power flow changes in the tie-line $(i,j)$ connecting markets $a$ and $a'$ (further discussions on the optimal value of $\delta_{a',(i,j)}$ is provided in Section \ref{3_Benchmark}).
The constraint \eqref{eq::decen_load_bal_0} ensures that the nodal injections are greater than the nodal net-demand forecasts, the power generation and generation ramping of units are maintained within limits respectively through \eqref{eq::decen_Gen_Lim} and \eqref{eq::decen_Ramp_Lim}, the power flows of transmission lines are calculated and constrained in \eqref{eq::decen_L_flow}, the tie-line power flows are calculated in \eqref{eq::decen_T_flow}, and the chance-constrained region power balance is presented in \eqref{eq::decen_load_bal_5}.
As the nodal net-demands in region $a$ are random variables, constraint \eqref{eq::decen_load_bal_5} requires total supply must meet total net-demand with $1-t_a$ confidence level and $t_a \in [0,1]$.
The slack bus is deemed to be the first bus of area $1$, without loss of generality, and the associated voltage phase angle is set to zero in \eqref{eq::decen_theta_1}.

With information on the tie-line capacity price $\mu_{(i,j)}$, the willingness to pay for flexibility $\delta_{a',(i,j)}$ and the angles $\theta_{a',(i,j)}$ from neighboring areas $a'$, and the internal costs of generation $C_{a,g}$, the market operator aims to meet the net-demand in the intraday market at the minimum cost possible, while respecting a reliability standard (as modeled by chance constraint \eqref{eq::decen_load_bal_5}).

\vspace{-6pt}
\subsection{Intraday Market Coupling Game}

With fixed capacity prices represented by vector $\bm{\mu}$ (per unit flow) for all tie-lines, the interaction between market operators can be modeled as a non-cooperative game with {\em coupled} constraints. The elements of such game are:
\begin{itemize}
    \item A finite set of players $\mc{A}$, i.e. the market operators.
    \item A payoff function for each player, i.e. $V_a$ the objective function for intraday clearing market $a$ in  \eqref{eq::decen_obj_fun}.
    \item A set of coupled constraints \eqref{eq::decen_load_bal_0}-\eqref{eq::decen_theta_1} that define the set of feasible intraday market clearing outcomes.
    \item Willingness to pay for adjustments in interconnection flows over tie-line connecting buses $(i,j)$, i.e. $\delta_{a,(i,j)}$.
\end{itemize}
\vspace{-2pt}
\begin{definition}
Given the capacity prices $\bm{\mu}$, a Nash equilibrium of the intraday market coupling game capacity prices is a strategy profile 
$(\bm{\Delta} \mathbf{P}_a^*, \bm{\Delta} \mathbf{T}_a^*, \bm{\theta}_{a}^*)$ for each $a \in \mathcal{A}$ satisfying constraints \eqref{eq::decen_load_bal_0}-\eqref{eq::decen_theta_1} (given $\bm{\theta}^*_{-a}$) and willingness to pay for adjustments in interconnection flows $\bm{\delta}_{a}^*$ such that:
\begin{align*}
& V_a(
   \boldsymbol{\Delta} \mathbf{P}_a^*, 
   \boldsymbol{\Delta} \mathbf{T}^{*}_a,  
   \boldsymbol{\theta}^{*}_{a};  \boldsymbol{\delta}^{*}_{-a}, \boldsymbol{\mu} ) 
    \leq  V_a(
   \boldsymbol{\Delta} \mathbf{P}_a, \boldsymbol{\Delta} \mathbf{T}_a,  \boldsymbol{\theta}_{a};  \boldsymbol{\delta}^{*}_{-a}, \boldsymbol{\mu} )
\end{align*}
for all solutions $(\bm{\Delta} \mathbf{P}_a, \bm{\Delta} \mathbf{T}_a, \bm{\theta}_{a})$ that also satisfy constraints \eqref{eq::decen_load_bal_0}-\eqref{eq::decen_theta_1} (given $\bm{\theta}^*_{-a}$).

\end{definition}

\vspace{-6pt}
\subsection{Decentralized Iterative Mechanism for Trading Flexibility}
\label{5_Converg}
In this section we propose a decentralized iterative mechanism for flexibility trading between different regions. Our main convergence result implies that the mechanism identifies a Nash equilibrium of the intraday market coupling game.

The basic steps of the mechanism are as follows--see Algorithm \ref{Algorithm} for details. At each iteration $k > 0$ of the mechanism,

\begin{enumerate} 
	\item ({\em Information exchange}) Each area $a' \in \mathcal{A}$ reports the current terms of trade: $\delta^{k-1}_{a',(i,j)}$ and angles $\theta^{k-1}_{a',j}$ for each interconnection.
	\item ({\em Updating Tie-line Flows}) Given a capacity price $\mu_{(i,j)}^{k-1}$, each area $a \in \mc{A}$ solves problem (1) to identify the market clearing solution $\widehat{\bm{\Delta}\mathbf{T}}_a^k$, and willingness to pay $\widehat{\bm{\delta}}_{a}^k$.
	After calculating $\widehat{\bm{\Delta}\mathbf{T}}_a^k$, $\widehat{\bm{\theta}}_a^k$ and  $\widehat{\bm{\delta}}_{a}^k$, the market operator updates the flows, angles and selling prices along the tie-lines according to the following:
	\begin{align}
	& \hspace{-11pt} \bm{\Delta} \mathbf{T}_a^k = (1 - \rho_k) \bm{\Delta} \mathbf{T}_a^{k-1} + \rho_k \widehat{\bm{\Delta}\mathbf{T}}_a^k \label{eq::upT} \\
	& \bm{\theta}_a^k = (1 - \rho_k) \bm{\theta}_a^{k-1} + \rho_k \widehat{\bm{\theta}}_a^k \label{eq::upThta} \\
	& \bm{\delta}_a^k = (1 - \rho_k) \bm{\delta}_a^{k-1} + \rho_k \widehat{\bm{\delta}}_{a}^k \label{eq::upDlta}
	\end{align}
	with $\rho_k \rightarrow 0^+$, $\sum_k \rho_k = +\infty$ and $\sum_k \rho_k^2 < +\infty$.
	\item ({\em Updating Tie-line Capacity Prices})
	\begin{align}
	& \hspace{-10pt} \mu_{(i,j)}^k = \max \{ \mu_{(i,j)}^{k-1} \label{eq::muUpd} \\
	& + \beta ( \frac{|\Delta T_{a,(i,j)}^k| + |\Delta T_{a',(j,i)}^k|}{2} + |T^\mathrm{DA}_{a,(i,j)}|  - \bar{T}_{a,(i,j)} ), 0  \} \nn
	\end{align}
	where $\beta \in (0,1)$ and $\mu_{(i,j)}^0 = 0$, and $\bar{T}_{a,(i,j)}$ is the maximum tie-line power capacity.
\end{enumerate}

\begin{algorithm}
	\begin{algorithmic}[1]
		\caption{Mechanism for Intraday Market Coupling} \label{Algorithm}
		\renewcommand{\algorithmicrequire}{\textbf{Require:}}
		\REQUIRE Initialize $\bm{x}_a^0 = \left( \mathbf{\Delta T}_a^0, \bm{\theta}_a^0, \bm{\delta}_a^0 \right)$ and $\mu_{(i,j)}^0$ for all areas and tie-lines.
		\REQUIRE Maximum number of iterations $T+1$. Set $k=1$.
		\WHILE {$k \leq T+1$}
		\STATE Each area solves optimal dispatch (1) to obtain $\widehat{\mathbf{x}}_a^k = \left( \widehat{\bm{\Delta}\mathbf{T}}^{k}_a,
		\widehat{\bm{\theta}}_a^k,
		\widehat{\bm{\delta}}_a^k \right)$.
		\STATE Calculate $\bm{x}_a^k = (1-\rho_k) \bm{x}_a^{k-1} + \rho_k \widehat{\bm{x}}_a^k$.
		\STATE Calculate the capacity prices $\mu_{(i,j)}^k$ using \eqref{eq::muUpd} for all tie-lines.
		\STATE Exchanges $\bm{x}_a^k$ (related to the tie-lines) between neighboring areas.
		\STATE Set $k=k+1$.
		\ENDWHILE
	\end{algorithmic} 
\end{algorithm}
%
In updating the tie-line flows, angles, and terms of trade, each area uses an inertial update in \eqref{eq::upT}-\eqref{eq::upDlta} by weighting the current optimal solution, i.e. $\widehat{\bm{\Delta}\mathbf{T}}_a^k$, $\widehat{\bm{\theta}}_a^k$ and $\widehat{\bm{\delta}}_a^k$, with the previous step, rather than using the current solution from market clearing. The capacity price is updated with a constant step size.
The convergence analysis exploits this fast update of the capacity prices versus the slow update of tie-line flows and angles, and the willingness to pay for flexibility.
%


\vspace{-6pt}
\subsection{Convergence}

We now state the standing assumptions for the market clearing model in \eqref{eq::decen_obj_fun}-\eqref{eq::decen_theta_1}:

\begin{assumption} \label{assumption}
We make the following standing assumptions:
\begin{enumerate}
    \item The cost function $C_{a,g}(\cdot)$ is strictly convex, for all $g \in \mathcal{G}_a, a \in \mathcal{A}$.
    \item Each market can meet the local demand on its own without relying on the flexibility trades on the tie-lines.
\end{enumerate}
\end{assumption}
  
\begin{RK}
Since
the optimization problem in \eqref{eq::decen_obj_fun}-\eqref{eq::decen_theta_1} is convex, the primal and dual solutions are Lipschitz continuous (see e.g. \cite{LecParamOpt}).
\end{RK}

The steps we take to prove the convergence of the proposed iterative mechanism are as follows:
\begin{itemize}
    \item First, we prove convergence of the capacity prices $\mu_{(i,j)}^k \rightarrow \mu_{(i,j)}^\infty \geq 0$ (Lemma \ref{Lemma::MuConv}).
    \item Secondly, we prove that the sequence $\bm{x}_a^k \rightarrow \bm{x}_a^{\infty}$ converges (Theorem \ref{Theorem::XConv}) and that the limit values correspond to a Nash equilibrium of the market clearing game with capacity prices $\mu_{(i,j)}^\infty$ (Theorem \ref{Theorem::XConv}).
\end{itemize}
\begin{RK}
From Assumption \ref{assumption}.2 we conclude that there exists a maximum tie-line capacity price $\bar{\mu}_{(i,j)}>0$ for $(i,j) \in \mc{T}_a$ such that for $\mu_{(i,j)}>\bar{\mu}_{(i,j)}$ the optimal solution to \eqref{eq::decen_obj_fun}-\eqref{eq::decen_theta_1} is one with no flexibility trading along the tie-lines, i.e. $\widehat{\Delta T}_{a,(i,j)} = 0, \forall (i,j) \in \mc{T}_a$.
\end{RK}
First, we show that the capacity multipliers converge and complementary slackness holds in the limit.
\begin{lemma} \label{Lemma::MuConv}
	For every tie-line $(i, j)$ it holds that $\mu_{(i,j)}^k \rightarrow \mu_{(i,j)}^{\infty}\geq 0$ and 
	\begin{align}  \label{eq::lemma}
	\mu_{(i,j)}^k ( \frac{|\Delta T_{a,(i,j)}^k| + |\Delta T_{a',(j,i)}^k|}{2} + |T^\mathrm{DA}_{a,(i,j)}| - \bar{T}_{a,(i,j)} ) \rightarrow 0
	\end{align}
\end{lemma}
\begin{proof}
See Appendix \ref{App::LemmaMuConv}.
\end{proof}

We now state the main result of this section.

\begin{theorem} \label{Theorem::XConv}
    (a) For every market operator $a \in \mathcal{A}$, the sequence $\{\bm{x}_a^k=( \bm{\Delta}\mathbf{T}^{k}_a, \bm{\theta}_a^k, \bm{\delta}_a^k) : k\geq 0\}$ derived from the iterative mechanism in Algorithm \ref{Algorithm} 
    converges to $\bm{x}_a^{\infty}$. (b) The joint strategy $\bm{x}^{\infty}$ is a Nash equilibrium of the coupled intra-day market clearing game with capacity prices $\bm{\mu}^{\infty}$.
\end{theorem}
\begin{proof}
See Appendix \ref{App::LemmaXConv}.
\end{proof}

The use of two-time scales for implementing adjustments, i.e., {\em fast} intertie capacity pricing adjustments vs {\em slow} updates in intertie flows, is used to establish the above result. Intuitively, this can be seen as enabling the identification of approximately optimal capacity prices for every given set of intertie flows across areas.\footnote{ Similar two-time scale algorithms are used in reinforcement learning wherein {\em faster} value estimates (critic) are combined with {\em slower} policy updates (actor).} 

The following corollary states that {\em consensus} between market operators regarding the direction and magnitude of interconnection flows is an inherent feature of the equilibrium.

\begin{corollary}
For every intertie connecting node $i$ in area $a$ and node $j$ in area $a'$, it holds that:
\[
T^\mathrm{DA}_{a,(i,j)} + \Delta 
T_{a,(i,j)}^{\infty}
=-(
T^\mathrm{DA}_{a',(i,j)} + \Delta 
T_{a',(i,j)}^{\infty})
\]
\end{corollary}
\begin{proof}
From the update equations \eqref{eq::upT} and \eqref{eq::upThta}, and the DC flow \eqref{eq::decen_T_flow}, it follows that $
T^\mathrm{DA}_{a,(i,j)}+\Delta T_{a,(i,j)}^k
=\frac{\theta_{a,i}^k - \theta_{a',j}^{k-1}}{\bar{x}_{a,(i,j)}}$
and $
T^\mathrm{DA}_{a',(i,j)}+\Delta T_{a',(i,j)}^k
=\frac{\theta_{a',j}^k - \theta_{a,i}^{k-1}}{\bar{x}_{a,(i,j)}}
$. The result follows from taking limits as $k \rightarrow \infty$.
\end{proof}

\subsection{Black-box Model for Market Clearing} \label{sec:BlackBox}

We now consider a more general formulation of internal market clearing. Specifically, given angles $\bm{\theta}_{-a}$ of markets other than $a$, we denote by $\mathcal{F}_a(\bm{\theta}_{-a})$ the set of feasible internal market clearing solutions for market $a$.
Given the capacity prices $\bm{\mu}$ and prices for interconnection flows $\bm{\delta}_{-a}$, a more general formulation of market clearing for market $a$ is:
\begin{equation}
    \begin{array}{cr}
    \min & V_a(
   \boldsymbol{\Delta} \mathbf{P}_a, 
   \boldsymbol{\Delta} \mathbf{T}_a,  
   \boldsymbol{\theta}_{a};  \boldsymbol{\delta}_{-a}, \boldsymbol{\mu} ) \\
    \text{s.t.} & \\
     & (\bm{\Delta} \mathbf{P}_a, \bm{\Delta} \mathbf{T}_a, \bm{\theta}_{a}) \in \mathcal{F}_a(\bm{\theta}_{-a})
    \end{array}
\label{black_box}    
\end{equation}
The above formulation could be relevant when market operators differ in how their respective market clearing solutions address reliability concerns. In this more general setting we re-state assumption \eqref{assumption} as follows:
\begin{assumption}
Regularity conditions on problem \eqref{black_box}:
\begin{enumerate}
    \item The objective $V_a$ is strictly convex in $(\bm{\Delta} \mathbf{P}_a, \bm{\Delta} \mathbf{T}_a, \bm{\theta}_{a})$. For fixed values of $(\bm{\Delta} \mathbf{P}_a, \bm{\Delta} \mathbf{T}_a, \bm{\theta}_{a})$, the objective function $V_a$ is (strictly) monotone increasing in the capacity price $\mu_{(i,j)}$.
    \item The set-valued map $\mathcal{F}_a(\bm{\theta}_{-a})$ is continuous. For fixed $\theta_{-a}$, the set $\mathcal{F}_a(\bm{\theta}_{-a})$ is closed + convex.
    \item Each market can meet the local demand on its own without relying on the flexibility trades on the tie-lines.
\end{enumerate}
\end{assumption}
We now state as a corollary to extension of Theorem 1 to a general black-box setting.
\begin{corollary}
Under the assumption II.2. the iterative mechanism converges to a Nash equilibrium of the market coupling game under general convex market clearing model as in \eqref{black_box}. 
\end{corollary}

\vspace{-6pt}
\section{A Theoretical Benchmark for Market Efficiency Analysis}
\label{3_Benchmark}
%
%
%
%

In this section we define a theoretical benchmark as the solution to the optimization problem faced by a hypothetical omniscient market operator with complete information on the cost functions of {\em all} generators, the forecasts for the nodal demands, and the grid structures, and the reliability standards for all regions.
Then, we show that the Nash equilibrium identified in Theorem 1 corresponds to the solution of the optimal (centralized) procurement of flexibility subject to each  individual market's chance constraint.

\vspace{-6pt}

\subsection{Theoretical Benchmark}

The theoretical benchmark is formulated as:
\begin{flalign}
    & \min_{\bm{\Delta} \mathbf{P}_a,\bm{\Delta} \mathbf{T}_a, \bm{\theta}_{a}, \forall a \in \mc{A}} V = \sum_{a \in \mc{A}} \sum_{g \in \mc{G}_a} C_{a,g} (P^\mathrm{DA}_{a,g} + \Delta P_{a,g})  \label{eq::cen_obj_fun}  \\
    & \mbox{s.t.}    \nn \\
    & \mathbf{M}_{a} (\mathbf{P}^\mathrm{DA}_{a} + \bm{\Delta} \mathbf{P}_{a})-\mathbf{B}_{a} \bm{\theta}_{a} - \mathbf{R}_{a} (\mathbf{T}^\mathrm{DA}_{a} + \mathbf{\Delta T}_a) \succeq    \mathbf{D}_{a}, \nn \\
    & \hspace{160pt} (\bm{\alpha}_{a}), \hspace{4pt} \forall a \in \mc{A}, \label{eq::cen_load_bal_0}\\
        & \underline{\mathbf{P}}_a \leq \mathbf{P}^\mathrm{DA}_a + \bm{\Delta} \mathbf{P}_a \leq \overline{\mathbf{P}}_a, \hspace{35pt} (\bm\nu_{a} , \bm\lambda_{a}), \hspace{4pt} \forall a \in \mc{A}, \\
            &\underline{\mathbf{R}}_a \leq  \bm{\Delta} \mathbf{P}_a \leq \overline{\mathbf{R}}_a, \hspace{64pt} (\bm\psi_{a} , \bm\varphi_{a}), \hspace{4pt} \forall a \in \mc{A}, \label{eq::cen_Ramp_Lim}\\
             & - \overline{H}_{a,(i,j)} \leq \frac{\theta_{a,i} - \theta_{a,j}}{x_{a,(i,j)}} \leq \overline{H}_{a,(i,j)}, \hspace{6pt} (\kappa_{a,(i,j)}, \eta_{a,(i,j)}), \nn\\
             & \hspace{130pt} \forall (i,j) \in \mc{H}_a, \hspace{4pt}  \forall a \in \mc{A}, \\
                 & \Delta T_{a,(i,j)} = \frac{\theta_{a,i} - \theta_{a',j}}{\bar{x}_{a,(i,j)}} - T_{a,(i,j)}, \hspace{10pt} (\xi_{a,(i,j)}), \nn \\ 
    & \hspace{131pt} \forall (i,j) \in \mc{T}_a, \hspace{4pt} \forall a \in \mc{A}, 
 \end{flalign}
 \begin{flalign}        
    & -\overline{\mathbf{T}}_a \leq \mathbf{T}^\mathrm{DA}_a + \bm{\Delta} \mathbf{T}_a \leq \overline{\mathbf{T}}_a, \hspace{10pt} (\bar{\kappa}_{a,(i,j)}, \bar{\eta}_{a,(i,j)}), \nn \\
    & \hspace{131pt} \forall (i,j) \in \mc{T}_a, \hspace{4pt} \forall a \in \mc{A}, \label{eq::TcenMarkIneq} \\
& \Pr(\mathbf{1}^{\top} \mathbf{D}_a \hspace{-1pt} \leq
 \mathbf{1}^{\top} (\mathbf{P}^\mathrm{DA}_a + \bm{\Delta} \mathbf{P}_a) 
 -\mathbf{1}^{\top} (\mathbf{T}^\mathrm{DA}_a \hspace{-1pt} + \hspace{-1pt} \bm{\Delta} \mathbf{T}_a))\hspace{-1pt} \hspace{-1pt} \hspace{-1pt}
 \nn \\ &\hspace{110pt} \geq 1 - t_a, 
       \hspace{5pt} (\gamma_a), \hspace{4pt} \forall a \in \mc{A}, \label{eq::cen_load_bal_5} \\
    & \theta_{1,1} = 0.   \label{eq::cen_theta_1}
\end{flalign}
In the formulation, an omniscient agent aims to minimize total cost of power generation while satisfying the demand \eqref{eq::cen_load_bal_0} and reliability \eqref{eq::cen_load_bal_5} constraints of each area by making use of the tie-lines between areas. The solution to the above problem is denoted with $\{ \bm{\Delta} \mathbf{P}_{a}^{\ast}, \bm{\Delta} \mathbf{T}_{a}^{\ast}, \bm{\theta}_{a}^{\ast}\}, \hspace{5pt} \forall a \in \mc{A}$.

\vspace{-6pt}
\subsection{Efficient Decentralized Intraday Market Operation} 

In what follows we show that with the appropriate terms of trade from neighboring markets (i.e. values of capacity price $\mu_{(i,j)}$,
the willingness to pay for flexibility $\delta_{a',(i,j)}$ and the angles $\theta_{a',(i,j)}$), the solution to the intraday market clearing problem for market $a$ corresponds to the theoretical benchmark introduced in the previous section.

\begin{proposition}
 \label{Prop::Proposition_1}
If the terms of trade from market $a$ with neighboring areas $a' \in \mathcal{A}\backslash{a}$ are as follows: 
\begin{align}
    \frac{\mu_{(i,j)}}{2} \coloneqq (-\bar{\kappa}_{a,(i,j)}^{\ast} &+ \bar{\eta}_{a,(i,j)}^{\ast}) \cdot \mbox{sign}(\Delta T_{a,(i,j)}^{\ast}), \nn \\
    \delta_{a',(i,j)} &\coloneqq \gamma_{a'}^{\ast} + \alpha_{a',j}^{\ast}, \nn \\
    \theta_{a',j} &\coloneqq \theta_{a',j}^{\ast},
    \label{eq::tran_rates}
\end{align}
then the solution to the DC-OPF in \eqref{eq::decen_obj_fun}-\eqref{eq::decen_theta_1}, \textit{i.e.} $\{ \widehat{\bm{\Delta}  \mathbf{P}}_{a}, \widehat{\bm{\Delta}  \mathbf{T}}_{a}, \widehat{\bm{\theta}}_{a}\}, \hspace{5pt} \forall a \in \mc{A}$ corresponds to the benchmark $\{ \bm{\Delta} \mathbf{P}_{a}^{\ast}, \bm{\Delta} \mathbf{T}_{a}^{\ast}, \bm{\theta}_{a}^{\ast}\}, \hspace{5pt} \forall a \in \mc{A}$.
\end{proposition}

\begin{proof}
See Appendix \ref{App::Proposition_1}.
\end{proof}


The optimal willingness to pay for power flow changes in the tieline $(i,j)$ ($\delta_{a',(i,j)}$) is the sum of the prices associated with the chance constraint $\gamma_{a'}^*$ and the nodal price at node $j$ $\alpha_{a',j}^*$. 
While $\gamma_{a'}^*$ depends on the aggregate reliability of the region, $\alpha_{a',j}^*$ depends on the ability of the region to deliver extra unit of power to node $j$ which is limited by the transmission network. Because of this aggregate versus local distinction between $\gamma_{a'}$ and $\alpha_{a',j}$, we observe price for reliability $\gamma_{a'}$ demonstrates smoother changes during the intraday market coupling process compared to nodal prices (see Section \ref{6_Simul} and Fig. \ref{fig:simul1}).


In the main result of this section (Theorem 2) we show that the Nash equilibrium identified in Theorem 1 corresponds to the theoretical benchmark. The following lemma is needed in the proof of Theorem 2.

\begin{lemma} \label{lemma::Xfeasible}
    The limiting solution obtained from the iterative mechanism, \textit{i.e.} $\bm{x}_a^{\infty}, \forall a \in \mc{A}$ is feasible solution of the centralized OPF problem in \eqref{eq::cen_obj_fun} - \eqref{eq::cen_theta_1}.
\end{lemma}
\begin{proof}
See Appendix \ref{App::Xfeasible}.
\end{proof}

Next, we state and prove the main result of this section.
\begin{theorem} \label{theorem::KKtEquiv}
    Assuming that all regions adopt the chance-constraint formulation for their intraday market, as in \eqref{eq::decen_obj_fun} - \eqref{eq::decen_theta_1}, then 
    the final values derived from the iterative mechanism, \textit{i.e.} $\bm{x}_a^{\infty}, \forall a \in \mc{A}$, satisfy all the KKT conditions of the centralized OPF problem in \eqref{eq::cen_obj_fun} - \eqref{eq::cen_theta_1}.
\end{theorem}
\begin{proof}
See Appendix \ref{App::KKtEquiv}.
\end{proof}
%
\begin{RK}
A critical step in the proof of Theorem \ref{theorem::KKtEquiv} is to show that one half of the limiting capacity price $\mu^{\infty}_{(i,j)}$ corresponds to the shadow price of interconnection capacity between nodes $i$ and $j$, i.e:
\begin{align*}
\frac{\mu_{(i,j)}^{\infty}}{2} = | \bar{\eta}_{a,(i,j)}^{\ast}-\bar{\kappa}_{a,(i,j)}^{\ast}|& ~~~\mbox{if}~~|\Delta T_{a,(i,j)}^{\ast} + T^\mathrm{DA}_{a,(i,j)} | = \bar{T}_{a,(i,j)}
\end{align*}
and $\frac{\mu_{(i,j)}^{\infty}}{2}=0$ otherwise.
\end{RK}


\section{Simulations}
\label{6_Simul}
\begin{figure*}
\centering
    \subfloat[Nodal prices for the tie-lines\label{subfig-1:delta}]{%
       \includegraphics[width=0.32\linewidth]{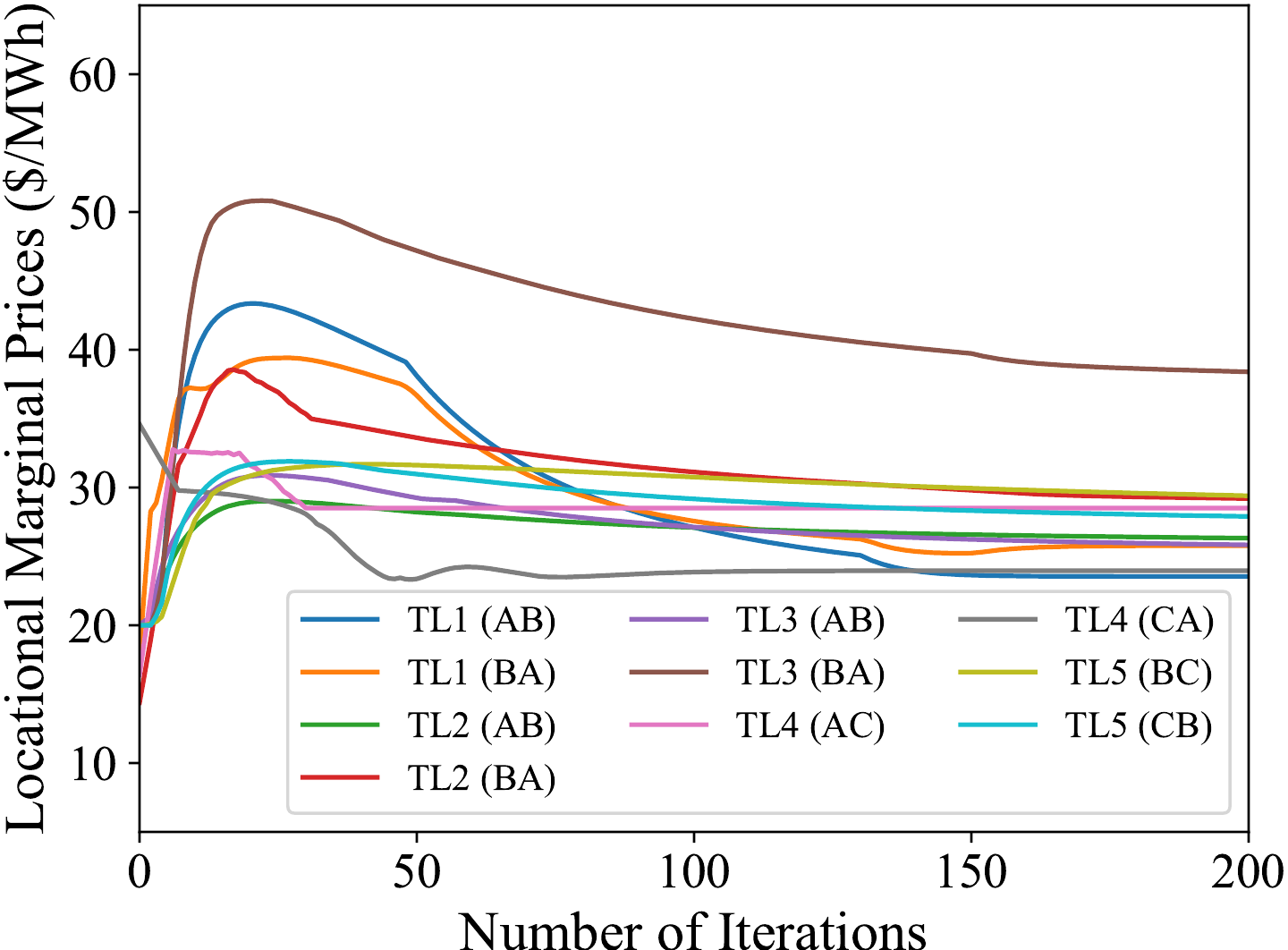}}
    \subfloat[Price for reliability ($\gamma_a$)\label{subfig-1:gamma}]{%
        \includegraphics[width=0.32\linewidth]{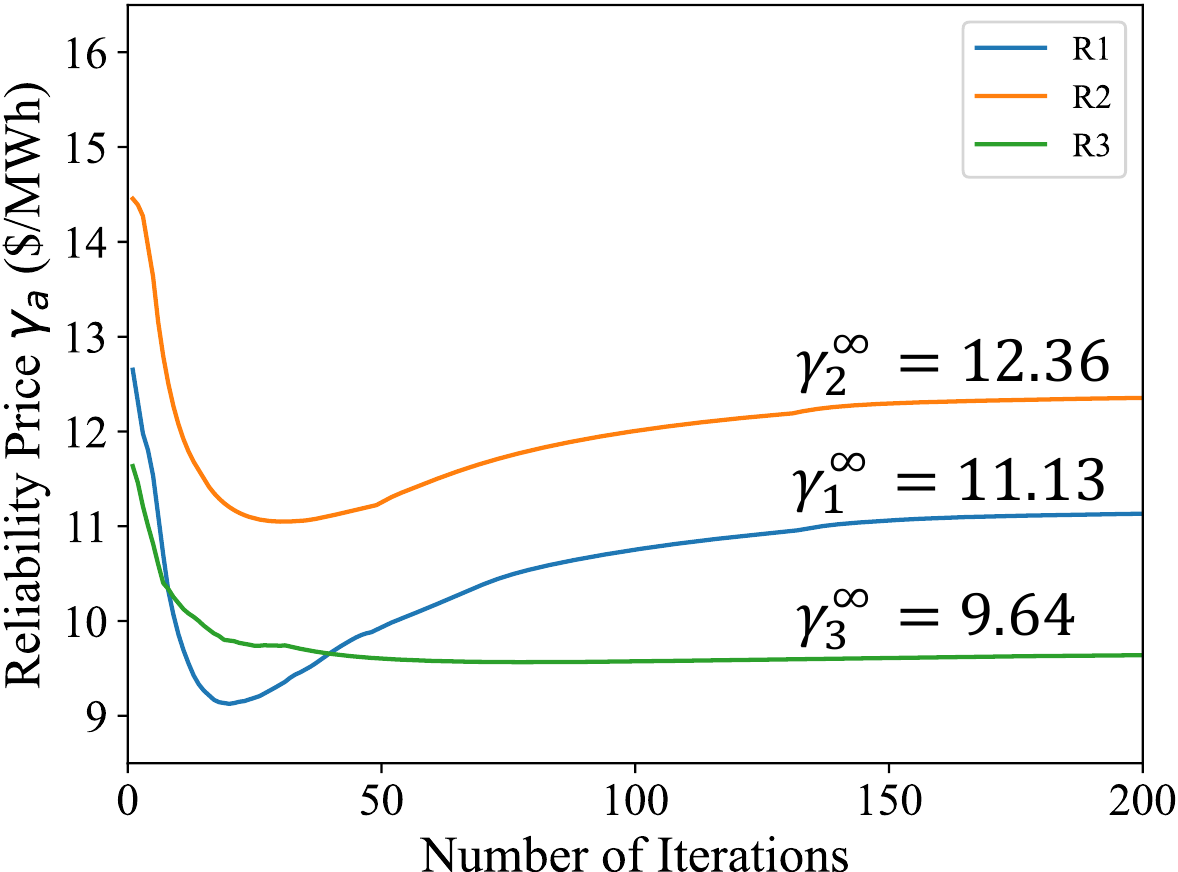}}
    \subfloat[Tie-line flows\label{subfig-1:T}]{%
        \includegraphics[width=0.32\linewidth]{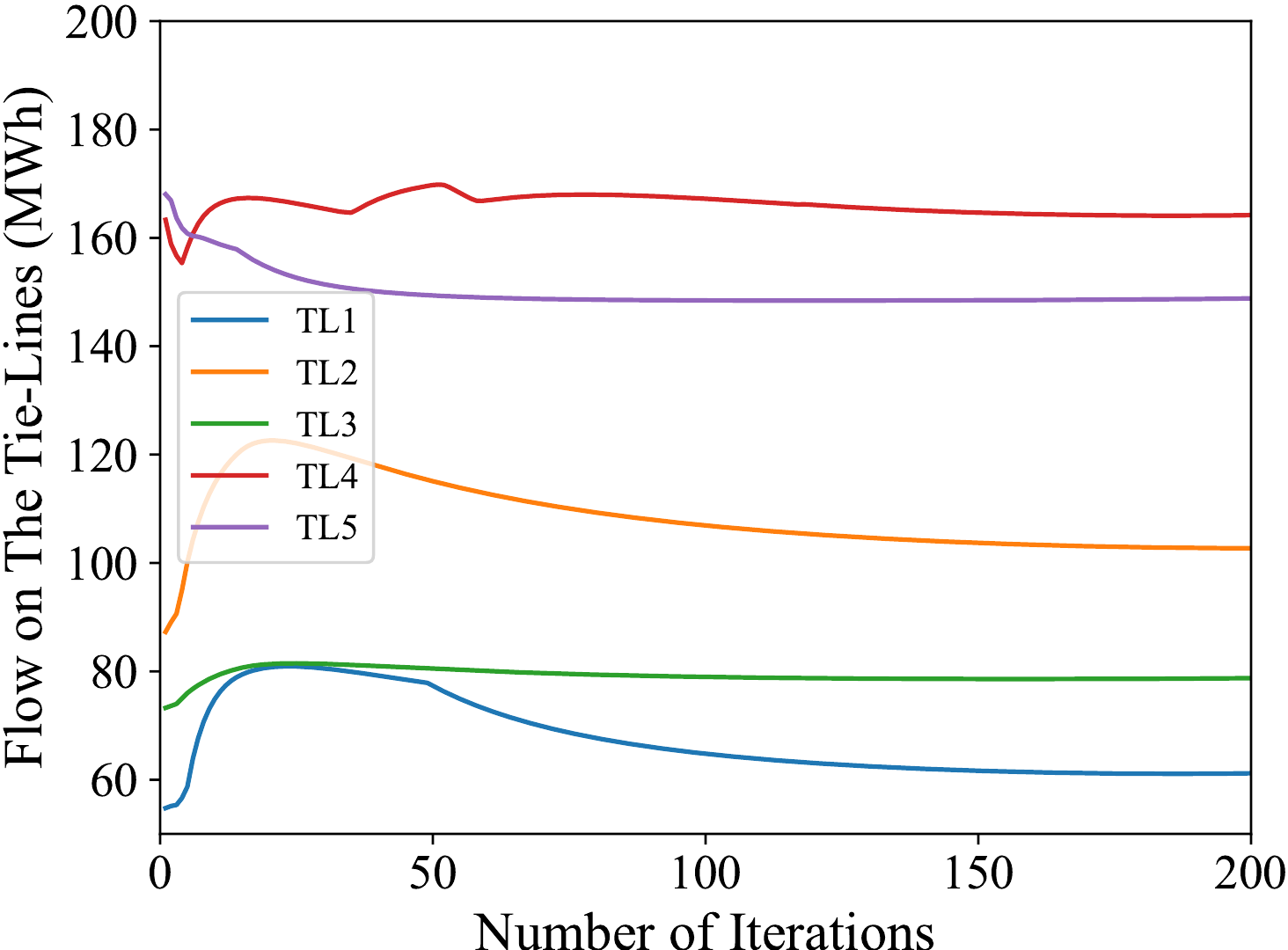}}
    \hfill
        \vspace{6pt}
    \caption{Nodal prices, Price for reliability ($\gamma_a$), and the tie-line flows.}
    \label{fig:simul1}
 \end{figure*}
 
 \begin{figure*}
\centering
    \subfloat[Nodal prices for the tie-lines\label{subfig-2:delta}]{%
       \includegraphics[width=0.32\linewidth]{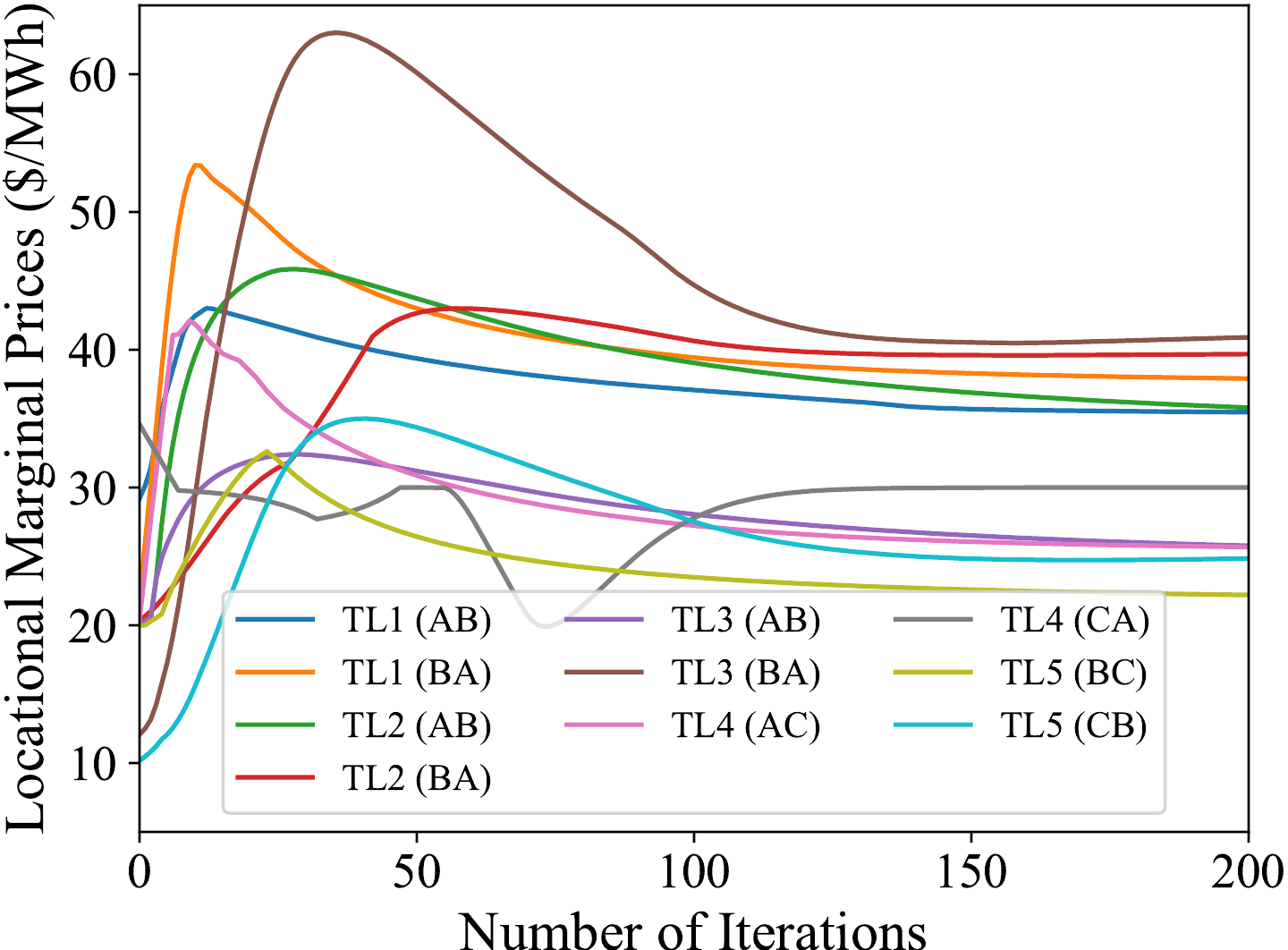}}
    \subfloat[Price for reliability ($\gamma_a$)\label{subfig-2:gamma}]{%
        \includegraphics[width=0.32\linewidth]{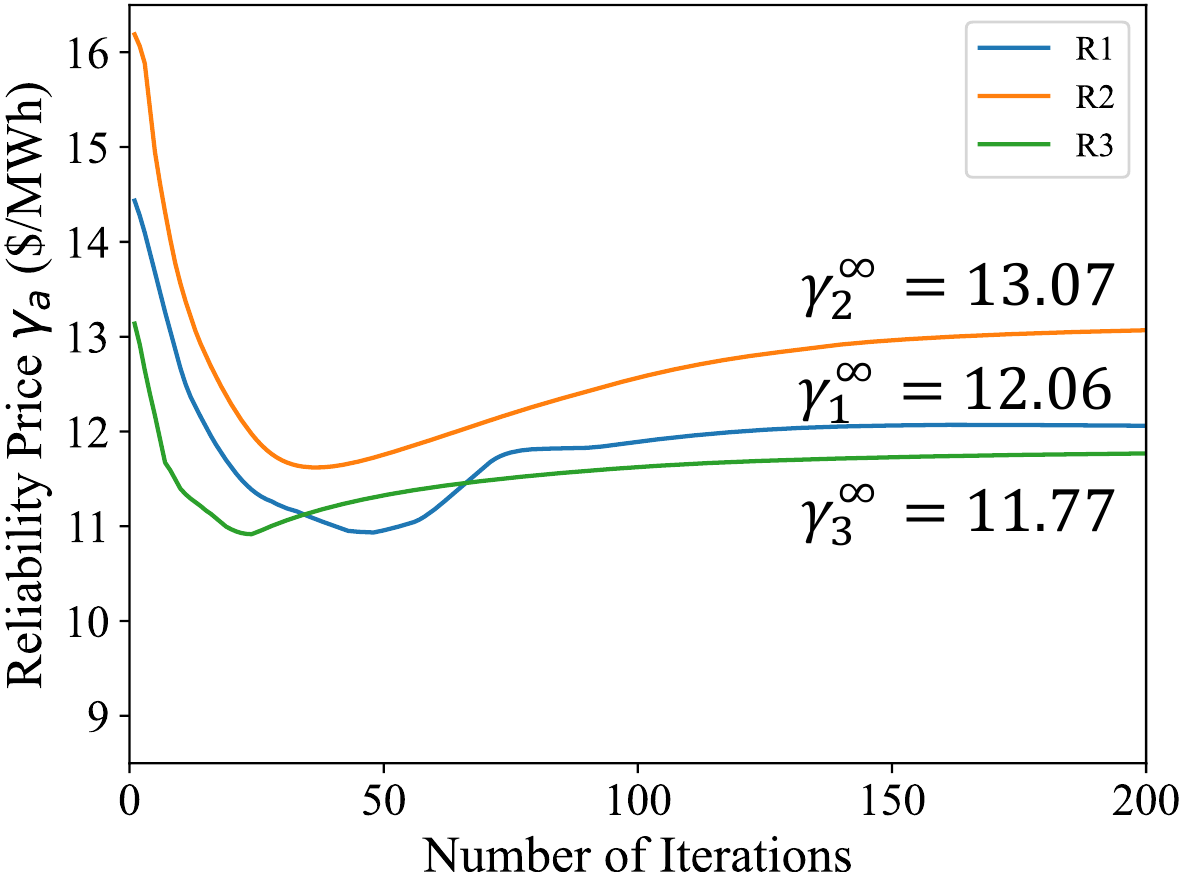}}
    \subfloat[Tie-line flows\label{subfig-2:T}]{%
        \includegraphics[width=0.32\linewidth]{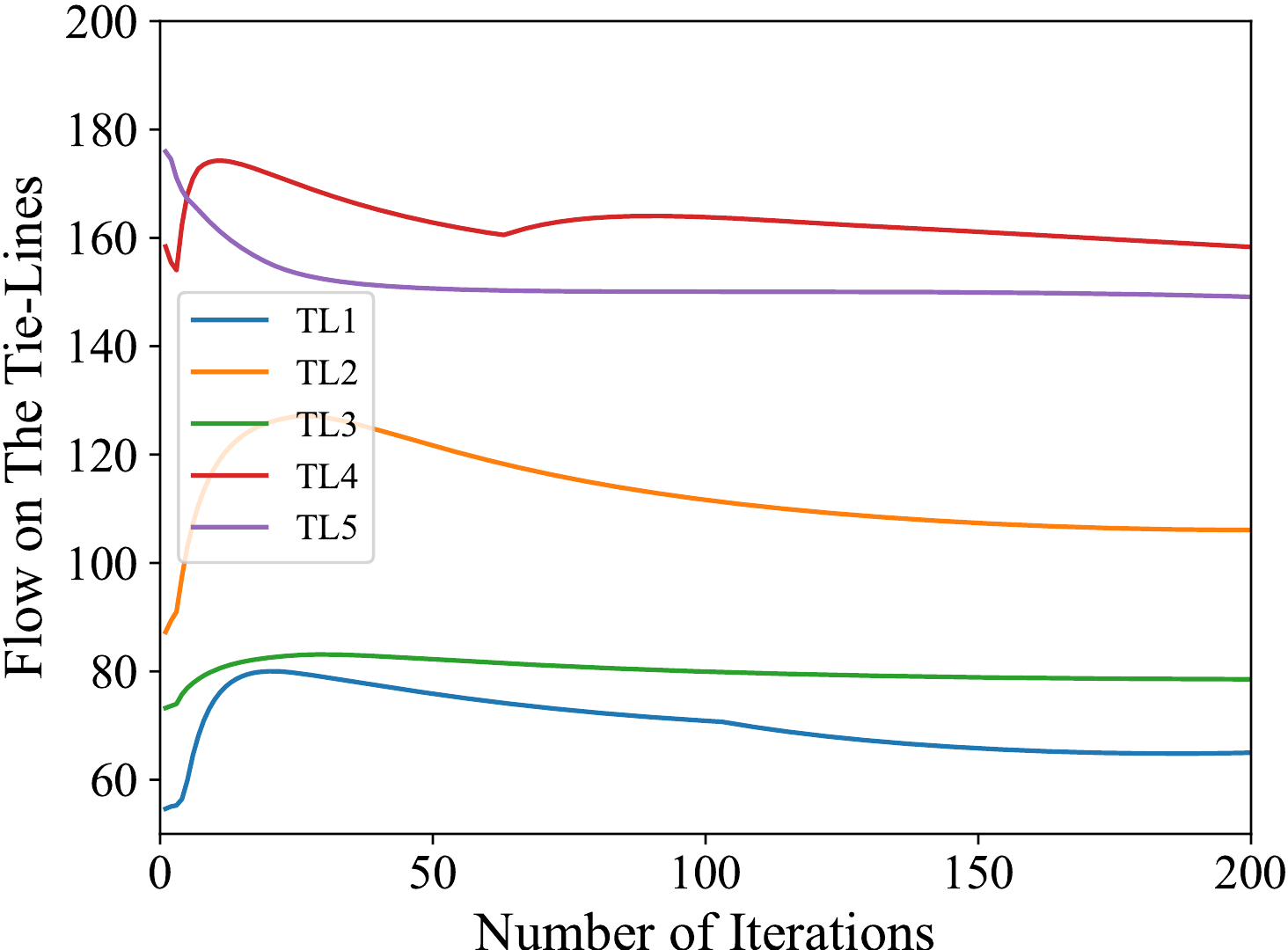}}
    \hfill
    \vspace{6pt}
  \caption{Nodal prices, Price for reliability ($\gamma_a$), and the tie-line flows  given tighter ramping rates.}
  \label{fig:simul2}
 \end{figure*}
 
 \begin{figure}
 \centering
       \includegraphics[width=0.90\linewidth]{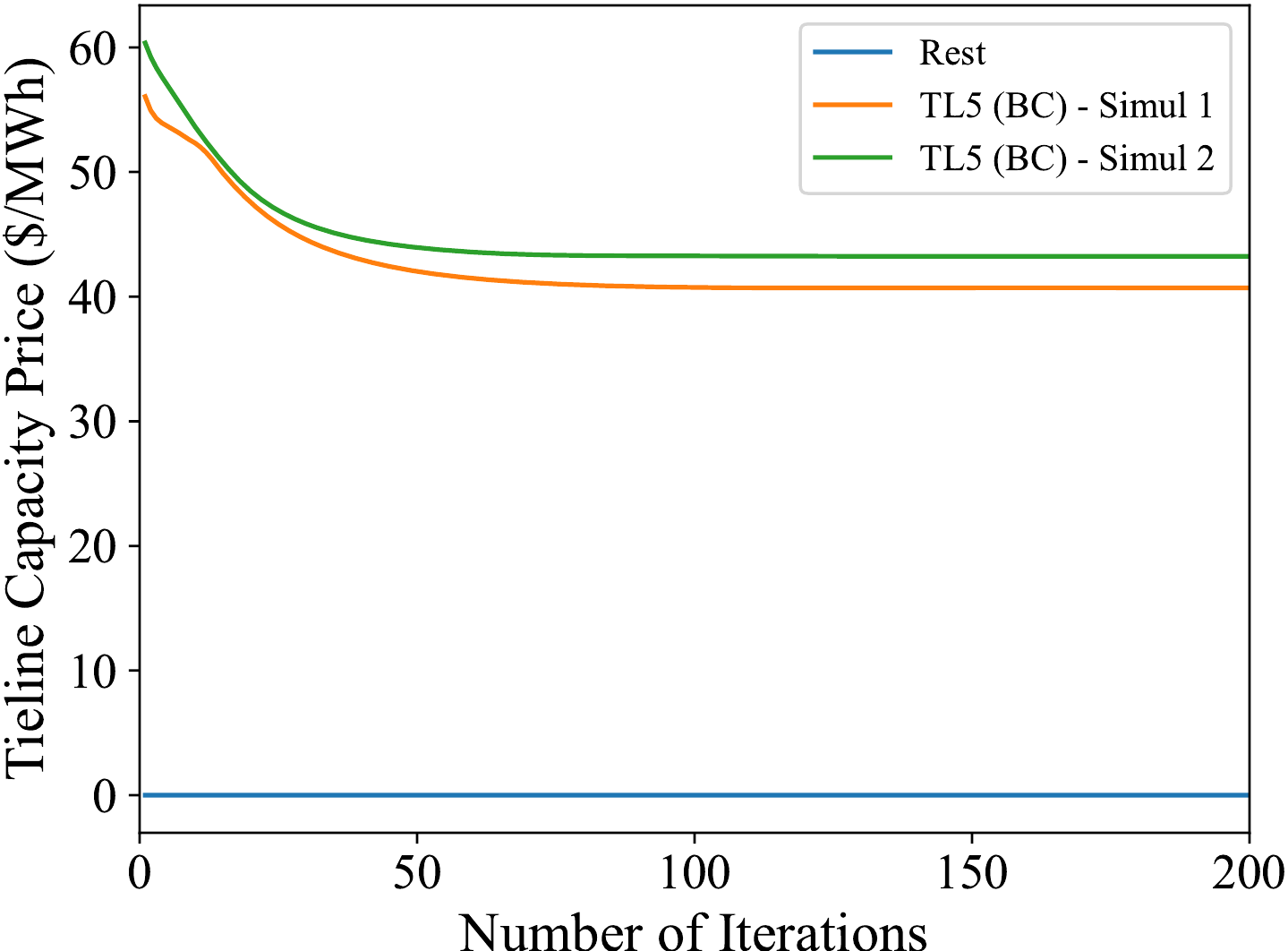}
       \caption{Capacity prices with ($\mbox{Simul}_2$) and without ($\mbox{Simul}_1$) tighter ramping rates}
    \label{fig:capPrice}
 \end{figure}

We use the IEEE Reliability Test System consisting of three zones, 73 buses, 115 transmission lines, 5 tie-lines, and 96 generators. See \cite{IEEE_RTS_1996} for the details of test system including the network topology, cost functions of the generators, and data for the nodal demands. We increase the capacity of all generators by 20\% to ensure the feasibility of this test case in the presence of net-demand uncertainty. The cost functions of the generators are extrapolated accordingly to reflect the extra capacities. 
We make the following changes to the testbed: 
\begin{itemize}
    \item In order to highlight the impact of (binding) tie-line capacity limits on capacity prices, we reduce the capacity of the tie-line between region B and C to 150MW.
    \item In order to highlight the impact of (binding) ramping constraints, we reduce the ramping rates of generators to one-third of their original values, because the total amount of the scheduled resources in this market is considerably smaller than that of the DA market, and the effects of ramping rates on prices manifest.
\end{itemize}

We assume that the real-time nodal net-demands of each region, i.e., $\mathbf{D}_a$, follow a normal distribution with the mean equal to DA forecast (represented by nodal demands given in \cite{IEEE_RTS_1996}) and the coefficient of variance equal to 6\%. 
Also, the confidence level of all areas are deemed 95\%, i.e., $t_a=0.05, \, a\in\mathcal{A}$.
We first schedule the generators for the DA load forecasts in the DA market \cite{KirchenBook}, and then use the proposed intra-day market to meet the reliability requirements of each region.
We assume that all the generators can participate in both the DA market and the intra-day market. The limitations of the slow-ramping generators to provide services in the intra-day market are reflected in their ramping rates.

Figs. \ref{fig:simul1}(a)-(c) show the nodal prices, price for reliability ($\gamma$), tie-line flows over the iterations. Initially, most of the nodal prices increase sharply before settling down to the optimal price values (Fig. \ref{fig:simul1}(a)). In contrast, the prices for reliability $\gamma_a$ decrease in the early stages of the iterations as the regions greedily try to optimize their market performances. As the number of iterations increases, the price for reliability gradually increases converging to a feasible optimal value (Fig. \ref{fig:simul1}(b)). We note that nodal prices tend to have larger swings over the iterations compared to the price for reliability. This is because prices for reliability measure the aggregate potential of a region to meet the changes in the chance constraint. The final values for $\{\gamma_a\}_{a\in\mathcal{A}}$ is smaller than the initial values, indicating that the trade between the regions leads to  reduced regional costs for reliability.


%

%

 The capacity price for the tie-line between regions B and C is nonzero, indicating that this line is congested (Fig. \ref{fig:capPrice}). We note that the capacity limits on the tie-lines are not hard coded in the proposed intraday model, rather it is guaranteed by the competition between the regions and the capacity prices. Hence, at the early steps of the simulation, we have higher flows on the tie-line between regions B and C than its actual capacity, resulting in higher capacity prices. As iterations continue, the mechanism for updating tie-line capacity prices in \eqref{eq::muUpd} decreases the flow on this line, and correspondingly decreases the capacity price. For the remaining tie-lines that are not congested, the  capacity prices are zero.

In order to see the effect of the ramping constraints on the market prices and intertie flows, we consider a second scenario in which the ramping rates of the generators are reduced by 25\%. In this scenario, the prices and tie-line flows follow a similar convergence path to the initial scenario (Fig. \ref{fig:simul2}). As shown in Figs. \ref{fig:capPrice}, \ref{subfig-1:gamma} and \ref{subfig-2:gamma}, the tighter ramping rates lead to increases in the final values of the capacity $\mu^a_{i,j}$ and reliability ($\gamma_a^\infty$) prices. 

 %
 %
 
%

\section{Conclusion}
\label{7_Conclus}

A decentralized market mechanism for flexibility margin trading is proposed. Regional market operators participate in this decentralized market to benefit from trading flexibility with neighboring market operators while maintaining privacy. In the proposed mechanism, areas engage in an iterative mechanism where neighboring regions share limited information on the tie-lines with each other and update their economic dispatch accordingly. Analytical results show the convergence of the proposed iterated mechanism to the Nash equilibrium of the intraday market coupling game. The IEEE-RTS-96 test system with three regional markets is used to perform the numerical experiments.  Numerical results validate the performance of the proposed decentralized mechanism.

\appendices

\section{The Proof of Lemma \ref{Lemma::MuConv}}
\label{App::LemmaMuConv}
Suppose the sequence $\{\mu_{(i,j)}^k : k>0\}$ does not converge. This means that one of the following two cases occur: either $(i)$, $\mu_{(i,j)}^k \rightarrow +\infty$, or $(ii)$, this sequence oscillates with $\lim \inf \mu_{(i,j)}^k < \lim \sup \mu_{(i,j)}^k$. Case $(i)$ is easily discarded since from the assumption	on maximum tie-line capacity price we have:
	\begin{align}
	\mu_{(i,j)}^k \rightarrow +\infty \Rightarrow |\Delta T_{a,(i,j)}^k| + |\Delta T_{a',(j,i)}^k| \rightarrow 0
	\end{align}
	which means that the sequence $\{\mu_{(i,j)}^k \}$ converges to a single number, c.f. \eqref{eq::muUpd}, which is a contradiction. Now let us consider case $(ii)$. Capacity price oscillation implies that for some $\epsilon > 0$
	\begin{align}
	& \lim \sup_{k} \{  \frac{|\Delta T_{a,(i,j)}^k| + |\Delta T_{a',(j,i)}^k|}{2} + |T^\mathrm{DA}_{a,(i,j)}| - \bar{T}_{a,(i,j)} \} \geq \epsilon \nonumber \\
	& \hspace{0pt} \geq \lim \inf_{k} \{    \frac{|\Delta T_{a,(i,j)}^k| + |\Delta T_{a',(j,i)}^k|}{2} + |T^\mathrm{DA}_{a,(i,j)}| - \bar{T}_{a,(i,j)}  \}
	\end{align}
	Note that if the sequence $\{\mu_{(i,j)}^k \}$ is oscillating, then the $\lim \inf_{k} \{  \frac{|\Delta T_{a,(i,j)}^k| + |\Delta T_{a',(j,i)}^k|}{2} + |T^\mathrm{DA}_{a,(i,j)}| - \bar{T}_{a,(i,j)}  \} \leq 0$, otherwise there exists a $k_0$ such that $\forall k \geq k_0,   \frac{|\Delta T_{a,(i,j)}^k| + |\Delta T_{a',(j,i)}^k|}{2} + |T^\mathrm{DA}_{a,(i,j)}| - \bar{T}_{a,(i,j)}  > 0$, which from  \eqref{eq::muUpd} it means that $\mu_{(i,j)}^k \rightarrow +\infty$ which is in contradiction with the assumption of the sequence $\{\mu_{(i,j)}^k\}$ is oscillating.
	\\
	Oscillation of the sequence $\{\mu_{(i,j)}^k\}$  means that the sequence $\{ \frac{|\Delta T_{a,(i,j)}^k| + |\Delta T_{a',(j,i)}^k|}{2} + |T^\mathrm{DA}_{a,(i,j)}| - \bar{T}_{a,(i,j)} \}$ exhibits infinitely many upcrossings of $\frac{\epsilon}{4}$ and $\frac{3\epsilon}{4}$. Let $k > 0$ denote an index for an iteration in which an upcrossing of $\frac{\epsilon}{4}$ takes place, i.e.:
	\begin{align}
	& \frac{|\Delta T_{a,(i,j)}^k| + |\Delta T_{a',(j,i)}^k|}{2} + |T^\mathrm{DA}_{a,(i,j)}| - \bar{T}_{a,(i,j)} \geq \frac{\epsilon}{4} \mbox{, and } \nonumber \\
	& \frac{|\Delta T_{a,(i,j)}^{k-1}| + |\Delta T_{a',(j,i)}^{k-1}|}{2} + |T^\mathrm{DA}_{a,(i,j)}| - \bar{T}_{a,(i,j)} < \frac{\epsilon}{4}
	\end{align}
	and let the $\tau(k)$ denote minimum number of steps that require to upcross $\frac{3\epsilon}{4}$, meaning that
	\begin{align}
	& \frac{|\Delta T_{a,(i,j)}^{k+\tau(k)}| + |\Delta T_{a',(j,i)}^{k+\tau(k)}|}{2} + |T^\mathrm{DA}_{a,(i,j)}| - \bar{T}_{a,(i,j)} \geq \frac{3\epsilon}{4}, \nn \\
	& \mbox{and } \nn \\
	& \frac{|\Delta T_{a,(i,j)}^{k+\tau(k)-1}| + |\Delta T_{a',(j,i)}^{k+\tau(k)-1}|}{2} + |T^\mathrm{DA}_{a,(i,j)}| - \bar{T}_{a,(i,j)} < \frac{3\epsilon}{4} \nn
	\end{align}
We are going to reject the assumption that the sequence $\{\mu_{(i,j)}^k\}$ is oscillating as follows: First, we show that $\tau(k)$ is an increasing function of $k$. Then we show that if $\tau(k)$ would be large enough, then the capacity price $\mu_{(i,j)}^{k+l}$ would be large enough (for some $0 \leq l \leq \tau(k)$) to make $\widehat{\Delta T}_{a,(i,j)}^{k}$ zero, which means that the sequence $\{  \frac{|\Delta T_{a,(i,j)}^k| + |\Delta T_{a',(j,i)}^k|}{2} + |T^\mathrm{DA}_{a,(i,j)}| - \bar{T}_{a,(i,j)} \}$ would not upcross $\frac{3\epsilon}{4}$.
	Since $\Delta T_{a,(i,j)}^{k+1} - \Delta T_{a,(i,j)}^k = \rho_k (\Delta \widehat{T}_{a,(i,j)}^{k+1} - \Delta T_{a,(i,j)}^k)$, it follows that
	\begin{align}
	&\frac{3\epsilon}{4} - \frac{\epsilon}{4} \leq \Delta T_{a,(i,j)}^{k+\tau (k)} - \Delta T_{a,(i,j)}^k \nn \\
	& \hspace{30pt} = \sum_{l=1}^{\tau (k)} (\Delta T_{a,(i,j)}^{k+l} - \Delta T_{a,(i,j)}^{k+l-1}) \nonumber \\
	& \hspace{30pt} = \sum_{l=1}^{\tau (k)} \rho_{k+l} (\Delta \widehat{T}_{a,(i,j)}^{k+l} - \Delta T_{a,(i,j)}^{k+l}) \nn  \leq 2 U_{a,(i,j)} \sum_{l=1}^{\tau (k)} \rho_{k+l} 
	\end{align}
	where $U_{a,(i,j)}$ is a maximum cap on the values of $\{\Delta T_{a,(i,j)}^k\}, \forall a \in \mc{A}, \forall (i,j) \in \mc{T}_a, \forall k$. One candidate for this maximum cap is the total load of the markets. Since $\{\rho_k\}$ is a decreasing sequence, we have $ \sum_{l=1}^{\tau (k)} \rho_{k+l} < \tau (k) \cdot \rho_k$. Hence, 
	\begin{align}
	\frac{\epsilon}{2} \leq 2 U_{a,(i,j)} \sum_{l=1}^{\tau (k)} \rho_{k+l}  \hspace{0pt} < 2 U_{a,(i,j)} \cdot \tau (k) \cdot \rho_k
	\end{align}
	which means that $\tau (k) > \frac{\epsilon}{4 \rho_k \cdot U_{a,(i,j)}}$. This gives a lower bound on the $\tau (k)$. We note that since $\rho_k \rightarrow 0^+$, $\frac{\epsilon}{4 \rho_k \cdot (\bar{T} - |T_{a,(i,j)}|)} \rightarrow +\infty$, which means that $\tau (k) \rightarrow +\infty$. This ends the first part of the proof. Now, note that from \eqref{eq::muUpd} we have that
	\begin{align} \label{eq::muLower}
	\mu_{(i,j)}^{k+z} \geq \mu_{(i,j)}^k + \beta z \frac{\epsilon}{4}, \hspace{10pt} 0 \leq z \leq \tau (k)
	\end{align}
	The intuition behind \eqref{eq::muLower} is that $ \frac{\epsilon}{4} \leq  \frac{|\Delta T_{a,(i,j)}^k| + |\Delta T_{a',(j,i)}^k|}{2} + |T^\mathrm{DA}_{a,(i,j)}| - \bar{T}_{a,(i,j)} \leq \frac{3\epsilon}{4}, \forall z \mbox{ s.t. } 0 \leq z \leq \tau (k)$. Hence, for each iteration, the value of $\mu_{(i,j)}$ increases at least $\beta \frac{\epsilon}{4}$. (cf. \eqref{eq::muUpd}). So \eqref{eq::muLower} implies that if $k$ is large enough, then $\tau (k)$ would be large enough such that for some $0 < \upsilon < \tau (k)$ we would have $\mu_{(i,j)}^{k+\upsilon} > \bar{\mu}_{(i,j)}$, which means that $\widehat{\Delta T}_{a, (i,j)}^{k + l} = 0$ for all $\upsilon < l \leq \tau (k)$ (cf. Assumption \ref{assumption}). Consequently, from \eqref{eq::upT} we have that $|\Delta T_{a,(i,j)}^{k + l}| \leq |\Delta T_{a,(i,j)}^{k + \upsilon}|$ for all $\upsilon < l \leq \tau (k)$. This means that $\frac{|\Delta T_{a,(i,j)}^k| + |\Delta T_{a',(j,i)}^k|}{2} + |T^\mathrm{DA}_{a,(i,j)}| - \bar{T}_{a,(i,j)}$ will not upcross $\frac{3\epsilon}{4}$, which is in contradiction with our assumption. 

\section{The Proof of Theorem \ref{Theorem::XConv}}
\label{App::LemmaXConv}
\subsection{Preliminaries}
We first state and prove a result that will be used in the proof of Theorem 1.
\begin{lemma} \label{Lemma::Cconverg}
Let $\{y_k:k\geq 0\}$, $\{\rho_k:k\geq 0\}$ and $\{b_k:k \geq 0\}$ be three sequences such that:
\begin{align}
    & y_{k+1} \leq (1-\rho_k) y_k + b_k, & \label{eq::LemmaConvEq1} 
\end{align}
where $y_k \geq0, \rho_k\in (0,1)$ and $\sum_{k} \rho_k = \infty, \sum_{k} b_k < \infty$. It follows that $y_k \rightarrow 0$.
\begin{proof}
From \eqref{eq::LemmaConvEq1} we have:
\begin{align*}
    -y_0 \leq \sum_{k= 0}^T (y_{k+1} - y_{k}) & \leq -\sum_{k = 0}^T \rho_k y_k + \sum_{k = 0}^T b_k
\end{align*}
Since $\sum_{k} b_k < \infty$, the previous inequality implies $$\sum_{k} \rho_k y_k < \infty.$$ Finally, the condition $\sum_{k} \rho_k = \infty$ implies $y_k \rightarrow 0$.
\end{proof}
\end{lemma}
\subsection{Proof: (a) Convergence}
We show the sequence $\{\bm{x}_a^k\}$ is \textit{Cauchy} (and hence it has a limit say $\bm{x}_a^{\infty}$). From the update rule in \eqref{eq::upT} - \eqref{eq::upDlta}, we have:
\begin{align} \label{eq::Xkp1minusXk}
    &\bm{x}_a^{k+1} - \bm{x}_a^k = (1 - \rho_{k+1}) \bm{x}_a^k + \rho_{k+1} \bm{\hat{x}}_a^{k+1} \nn \\
    & \hspace{115pt} - (1 - \rho_k) \bm{x}_a^{k-1} - \rho_k \bm{\hat{x}}_a^k \nn \\
    & \hspace{46pt} = (1-\rho_k) (\bm{x}_a^k - \bm{x}_a^{k-1}) + \rho_k (\bm{\hat{x}}_a^{k+1} - \bm{\hat{x}}_a^k) \nn \\ 
    & \hspace{107pt} + (\rho_{k+1} - \rho_k) (\bm{\hat{x}}_a^{k+1} - \bm{x}_a^k)
\end{align}
As stated in Remark II.1, the primal and dual solutions to market clearing problem are Lipschitz continuous 
\cite{LipschitzPinho11, Robinson80}. Thus, there exists a $L_a > 0$ such that:
\begin{align}
    & || \bm{\hat{x}}_a^{k+1} - \bm{\hat{x}}_a^k || \leq \sum_{a' \in \mc{A}(a)} L_a || \bm{x}_{a'}^{k+1} - \bm{x}_{a'}^k || + L_a || \bm{\mu}_a^k - \bm{\mu}_a^{k-1} || \nn \\
    & \hspace{90pt} = \mc{O}(\rho_k) + L_a || \bm{\mu}_a^k - \bm{\mu}_a^{k-1} ||
\end{align}
We now show that $\norm{\bm{x}_a^{k+1} - \bm{x}_a^k} \rightarrow 0$. From \eqref{eq::Xkp1minusXk} we have:
\begin{align}
    & \norm{\bm{x}_a^{k+1} - \bm{x}_a^k} \leq (1-\rho_k) \norm{\bm{x}_a^k - \bm{x}_a^{k-1}} + \rho_k \norm{\bm{\hat{x}}_a^{k+1} - \bm{\hat{x}}_a^k} \nn \\ 
    & \hspace{103pt} + (\rho_{k+1} - \rho_k) \norm{\bm{\hat{x}}_a^{k+1} - \bm{x}_a^k}
\label{inequality}    
\end{align}
Hence, \eqref{inequality} can be written as $y_{k+1} \leq (1-\rho_k) y_k + b_k$ with $y_k  \triangleq \norm{\bm{x}_a^k - \bm{x}_a^{k-1}}$ and 
$$b_k \triangleq \rho_k \norm{\bm{\hat{x}}_a^{k+1} - \bm{\hat{x}}_a^k} + (\rho_{k+1} - \rho_k) \norm{\bm{\hat{x}}_a^{k+1} - \bm{x}_a^k}.$$ From Lemma \ref{Lemma::Cconverg}, we have $\norm{\bm{x}_a^{k+1} - \bm{x}_a^k} \rightarrow 0$. We now show $\norm{\bm{x}_a^{k+2} - \bm{x}_a^k} \rightarrow 0$. As before we write:
\begin{align}
& \bm{x}_{a}^{k+2}-\bm{x}_{a}^{k+1}=(1-\rho _{k+1})(%
\bm{x}_{a}^{k+1}-\bm{x}_{a}^{k}) \nn \\
& \hspace{20pt} +\rho _{k+1}(\bm{%
\hat{x}}_{a}^{k+2}-\bm{\hat{x}}_{a}^{k+1}) \nn 
 + (\rho _{k+2}-\rho _{k+1})(%
\bm{\hat{x}}_{a}^{k+2}-\bm{x}_{a}^{k+1})
\end{align}

Then%
\begin{align}
&\bm{x}_{a}^{k+2}-\bm{x}_{a}^{k} = (\bm{x}_{a}^{k+2}-\bm{x}_{a}^{k+1}) + (\bm{x}_{a}^{k+1}-%
\bm{x}_{a}^{k}) \nn \\
&\hspace{10pt} = (1-\rho _{k+1})(\bm{x}_{a}^{k+1}-\bm{x}_{a}^{k})+\rho
_{k+1}(\bm{\hat{x}}_{a}^{k+2}-\bm{\hat{x}}_{a}^{k+1}) \nn \\
& \hspace{12pt} +(\rho_{k+2}-\rho _{k+1})(\bm{\hat{x}}_{a}^{k+2}-\bm{x}_{a}^{k+1}) 
+ (1-\rho _{k})(\bm{x}_{a}^{k}-\bm{x}_{a}^{k-1})\nn\\
&\hspace{12pt} +\rho _{k}(%
\bm{\hat{x}}_{a}^{k+1}-\bm{\hat{x}}_{a}^{k})   +(\rho_{k+1}-\rho _{k})(\bm{\hat{x}}_{a}^{k+1}-\bm{x}_{a}^{k}) \nn \\
&\hspace{10pt} = (1-\rho _{k})(\bm{x}_{a}^{k+1}-\bm{x}_{a}^{k-1})+\rho
_{k+1}(\bm{\hat{x}}_{a}^{k+2}-\bm{\hat{x}}_{a}^{k+1}) \nn \\
&\hspace{10pt} +(\rho_{k+2}-\rho _{k+1})(\bm{\hat{x}}_{a}^{k+2}-\bm{x}_{a}^{k+1}) + \rho_{k}(\bm{\hat{x}}_{a}^{k+1}-\bm{\hat{x}}%
_{a}^{k}) \nn \\
& \hspace{10pt} +(\rho _{k+1}-\rho _{k})(\bm{\hat{x}}_{a}^{k+1}-\bm{%
x}_{a}^{k})+(\rho _{k}-\rho _{k+1})(\bm{x}_{a}^{k+1}-\bm{x}%
_{a}^{k})
\end{align}%
From which we obtain%
\begin{align}
\left\Vert \bm{x}_{a}^{k+2}-\bm{x}_{a}^{k}\right\Vert \leq
(1-\rho _{k})\left\Vert \bm{x}_{a}^{k+1}-\bm{x}%
_{a}^{k-1}\right\Vert +b_{k} \label{eq::bk2nd}
\end{align}%
where
\begin{align} 
b_{k} \triangleq& \rho _{k+1}\left\Vert \bm{\hat{x}}_{a}^{k+2}-\bm{%
\hat{x}}_{a}^{k+1}\right\Vert + \rho _{k}\left\Vert \bm{\hat{x}}_{a}^{k+1}-\bm{\hat{x}}%
_{a}^{k}\right\Vert \nn \\
& +(\rho _{k+2}-\rho _{k+1})\left\Vert 
\bm{\hat{x}}_{a}^{k+2}-\bm{x}_{a}^{k+1}\right\Vert  \nn \\
& +(\rho _{k+1}-\rho _{k})\left\Vert \bm{\hat{x}}%
_{a}^{k+1}-\bm{x}_{a}^{k}\right\Vert \nn +(\rho _{k}-\rho
_{k+1})\left\Vert \bm{x}_{a}^{k+1}-\bm{x}_{a}^{k}\right\Vert
\end{align}%
Let $y_k  \triangleq \{\norm{\bm{x}_a^{k+1} - \bm{x}_a^{k-1}}\}$. Here again, using Lemma \ref{Lemma::Cconverg} we have $y_k =\norm{\bm{x}_{a}^{k+2}-\bm{x}_{a}^{k}} \rightarrow 0$. A similar argument is used to
show $\norm{\bm{x}_{a}^{k+T}-\bm{x}_{a}^{k}}
\rightarrow 0$ for any $T<\infty $.

\subsection{Proof: (b) Limit is Nash Equilibrium}
As  $\boldsymbol{\theta}^{k}_{-a}
\rightarrow \boldsymbol{\theta}^{\infty}_{-a}$, $\boldsymbol{\delta}^{k}_{-a} \rightarrow \boldsymbol{\delta}^{\infty}_{-a}$ and 
$\boldsymbol{\mu}^{k} \rightarrow \boldsymbol{\mu}^{\infty}$, the Lipschitz continuity of the primal solution to market clearing problem implies the sequence $\{\bm{\hat{x}}_a^k:k>0\}$ also has a limit say $\bm{\hat{x}}_a^{\infty}$. 
\\
Further, since
$$
\bm{x}_a^k-\bm{\hat{x}}_a^{\infty} =
(1-\rho_k)(\bm{x}_a^{k-1}-\bm{\hat{x}}_a^{\infty})
+
\rho_k(\bm{\hat{x}}_a^{k}-\bm{\hat{x}}_a^{\infty}),
$$
it follows that 
$$
||\bm{x}_a^k-\bm{\hat{x}}_a^{\infty}|| \leq
(1-\rho_k)||\bm{x}_a^{k-1}-\bm{\hat{x}}_a^{\infty}||
+
\rho_k||\bm{\hat{x}}_a^{k}-\bm{\hat{x}}_a^{\infty}||
$$
and by Lemma 3 it follows that $||\bm{x}_a^k-\bm{\hat{x}}_a^{\infty}|| \rightarrow 0$ or equivalently $\bm{x}_a^{\infty}= \bm{\hat{x}}_a^{\infty}$.

Let $\mathcal{L}_a$ denote the Lagrangian of the market clearing problem in $(1)-(8)$. The clearing solution at market $a$ at iteration $k$ satisfies:
\begin{align*}
& \mathcal{L}_a(
   \widehat{\boldsymbol{\Delta} \mathbf{P}}^{k}_a, 
   \widehat{\boldsymbol{\Delta} \mathbf{T}}^{k}_a,  
   \widehat{\boldsymbol{\theta}}^{k}_{a};  \boldsymbol{\delta}^{k-1}_{-a},\boldsymbol{\mu}^{k-1} )  \\
   & \hspace{24pt} \leq  
  \mathcal{L}_a(
   \boldsymbol{\Delta} \mathbf{P}_a, \boldsymbol{\Delta} \mathbf{T}_a,  \boldsymbol{\theta}_{a};  \boldsymbol{\delta}^{k-1}_{-a},\boldsymbol{\mu}^{k-1} ) 
\end{align*}
Since complementary slackness holds in the limit we obtain:
\begin{align*}
& V_a(
   \widehat{\boldsymbol{\Delta} \mathbf{P}}^{\infty}_a, 
   \widehat{\boldsymbol{\Delta} \mathbf{T}}^{\infty}_a,  
   \widehat{\boldsymbol{\theta}}^{\infty}_{a};  \boldsymbol{\delta}^{\infty}_{-a},\boldsymbol{\mu}^{\infty} )  \\
   & \hspace{24pt}\leq  
  V_a(
   \boldsymbol{\Delta} \mathbf{P}_a, \boldsymbol{\Delta} \mathbf{T}_a,  \boldsymbol{\theta}_{a};  \boldsymbol{\delta}^{\infty}_{-a},\boldsymbol{\mu}^{\infty} ) 
\end{align*}
for any feasible solution $(\boldsymbol{\Delta} \mathbf{P}_a, \boldsymbol{\Delta} \mathbf{T}_a,  \boldsymbol{\theta}_{a})$ with angles $\boldsymbol{\theta}_{-a}^{\infty}$.

\vspace{-6pt}
\section{The Proof of Proposition \ref{Prop::Proposition_1}}

We start by writing the KKT optimality conditions for both the centralized OPF and the decentralized market. 

The first order conditions with respect to $\theta_{a,i}$, $\theta_{a',j}$ and $\Delta T_{a,(i,j)}$ for tie-line $(i,j) \in \mc{T}_a$ in the centralized DC-OPF problem in \eqref{eq::cen_obj_fun}-\eqref{eq::cen_theta_1} are:
\begin{align}
    & \sum_{k|(i,k) \in \mc{H}_a} \big\{ \frac{1}{x_{a,(i,k)}} \cdot \big( \alpha_{a,i}^{\ast} - \alpha_{a,k}^{\ast} -\kappa_{a,(i,k)}^{\ast} + \eta_{a,(i,k)}^{\ast} \big) \big\} \nn \\
    & \hspace{127pt} - \frac{1}{\bar{x}_{a,(i,j)}} \cdot \xi_{a,(i,j)}^{\ast} = 0, \label{eq::KKTcentrThetI} \\
    & \sum_{k|(j,k) \in \mc{H}_{a'}} \big\{ \frac{1}{x_{a',(j,k)}} \cdot \big( \alpha_{a',j}^{\ast} - \alpha_{a',k}^{\ast} - \kappa_{a',(j,k)}^{\ast} + \eta_{a',(j,k)}^{\ast} \big) \big\} \nn \\
    & \hspace{127pt} + \frac{1}{\bar{x}_{a,(i,j)}} \cdot \xi_{a,(i,j)}^{\ast} = 0, \label{eq::KKTcentrThetJ} \\
    & \alpha_{a,i}^{\ast} - \alpha_{a',j}^{\ast} + \xi_{a,(i,j)}^{\ast} - \bar{\kappa}_{a,(i,j)}^{\ast} + \bar{\eta}_{a,(i,j)}^{\ast} 
    + \gamma_a^{\ast} - \gamma_{a'}^{\ast} = 0 \label{eq::KKTcentrT}.
\end{align}
Similarly, the first order conditions with respect to $\theta_{a,i}$ and $\Delta T_{a,(i,j)}$ for tie-line $(i,j) \in \mc{T}_a$ in the decentralized DC-OPF problem are:
\begin{align}
    & \sum_{k|(i,k) \in \mc{H}_a} \big\{ \frac{1}{x_{a,(i,k)}} \cdot \big( \hat{\alpha}_{a,i} - \hat{\alpha}_{a,k} - \hat{\kappa}_{a,(i,k)} + \hat{\eta}_{a,(i,k)} \big) \big\} \nn \\
    & \hspace{127pt} - \frac{1}{\bar{x}_{a,(i,j)}} \cdot \hat{\xi}_{a,(i,j)} = 0, \\
    & \hat{\alpha}_{a,i} + \hat{\xi}_{a,(i,j)} - \alpha_{a',j}^{\ast} - \bar{\kappa}_{a,(i,j)}^{\ast} + \bar{\eta}_{a,(i,j)}^{\ast} 
    + \hat{\gamma}_a - \gamma_{a'}^{\ast} = 0.
\end{align}
And the first order conditions with respect to $\theta_{a',j}$ and $\Delta T_{a',(j,i)}$ for tie-line $(j,i) \in \mc{T}_{a'}$ in the decentralized DC-OPF problem are:
\begin{align}
    & \sum_{k|(j,k) \in \mc{H}_{a'}} \big\{ \frac{1}{x_{a',(j,k)}} \cdot \big( \hat{\alpha}_{a',j} - \hat{\alpha}_{a',k} - \hat{\kappa}_{a',(j,k)} + \hat{\eta}_{a',(j,k)} \big) \big\} \nn \\
    & \hspace{124pt} - \frac{1}{\bar{x}_{a,(i,j)}} \cdot \hat{\xi}_{a',(j,i)} = 0, \\
    & \hat{\alpha}_{a',j} + \hat{\xi}_{a',(j,i)} - \alpha_{a,i}^{\ast} - \bar{\kappa}_{a',(i,j)}^{\ast} + \bar{\eta}_{a',(i,j)}^{\ast} 
    + \hat{\gamma}_{a'} - \gamma_{a}^{\ast} = 0.
\end{align}
We note that $\bar{\kappa}_{a',(i,j)}^{\ast}$ and $\bar{\eta}_{a',(i,j)}^{\ast}$ are the values (at the optimum) of the auxiliary dual variables for a different formulation of the centralized OPF in \eqref{eq::cen_obj_fun}-\eqref{eq::cen_theta_1} where we define the flow on the tie-line $(i,j)$ in reverse, \textit{i.e.} from area $a'$ to area $a$. One can immediately see that $\bar{\kappa}_{a',(i,j)}^{\ast} = -\bar{\kappa}_{a,(i,j)}^{\ast}$ and $\bar{\eta}_{a',(i,j)}^{\ast} = -\bar{\eta}_{a,(i,j)}^{\ast}$.
Now the following set of values for the parameters is the solution to the decentralized OPF, which are the same as the one for the centralized OPF:

\begin{tabular}{ p{0.25\textwidth} p{0.25\textwidth} }
$\hat{\alpha}_{a,i} = \alpha_{a,i}^{\ast},$ & $\hat{\alpha}_{a',j} = \alpha_{a',j}^{\ast},$ \\
$\Delta \hat{T}_{a,(i,j)}  = \Delta T_{a,(i,j)}^{\ast},$ & $\Delta \hat{T}_{a',(j,i)} = -\Delta T_{a,(i,j)}^{\ast},$ \\
$\hat{\xi}_{a,(i,j)} = \xi_{a,(i,j)}^{\ast},$ & $\hat{\xi}_{a',(j,i)} = -\xi_{a,(i,j)}^{\ast}$.
\end{tabular}
Hence, with the values of the transmission rates defined in \eqref{eq::tran_rates}, the market outcome of the decentralized intraday market in \eqref{eq::decen_obj_fun}-\eqref{eq::decen_theta_1} would be the same as the outcome of the centralized economic dispatch in \eqref{eq::cen_obj_fun}-\eqref{eq::cen_theta_1}.

\label{App::Proposition_1}

\vspace{-6pt}
\section{The Proof of Lemma \ref{lemma::Xfeasible}}
\label{App::Xfeasible}

Note that the solutions to the decentralized optimization problems in \eqref{eq::decen_obj_fun} - \eqref{eq::decen_theta_1} for each $a \in \mathcal{A}$ satisfy all the primal conditions of the centralized optimization in \eqref{eq::cen_obj_fun} - \eqref{eq::cen_theta_1} {\em except} the constraint \eqref{eq::TcenMarkIneq}. 
Here we show that the ultimate flows of the tie-lines, {i.e.}, $\Delta T_{a,(i,j)}^{\infty}$, derived from the iterative implementation of the decentralized mechanism also satisfy the constraint \eqref{eq::TcenMarkIneq}. 
From the fact that the sequence $\{\bm{x}_a^{k}\}$ converges, we know that
$$\Delta T_{a,(i,j)}^{\infty} = - \Delta T_{a',(j,i)}^{\infty} = \frac{\theta_{a,i}^{\infty}-\theta_{a',j}^{\infty}}{\bar{x}_{a,(i,j)}} - T^\mathrm{DA}_{a,(i,j)}.$$ Also, from the fact that the sequence $\{\mu_{(i,j)}^k\}$ converges we must have (cf. \eqref{eq::muUpd}):
\begin{align} \label{eq::boundT1}
    \frac{|\Delta T_{a,(i,j)}^{\infty}| + |\Delta T_{a',(j,i)}^{\infty}|}{2} + |T^\mathrm{DA}_{a,(i,j)}| - \bar{T}_{a,(i,j)} \leq 0
\end{align}
and since $\Delta T_{a,(i,j)}^{\infty} = - \Delta T_{a',(j,i)}^{\infty}$, \eqref{eq::boundT1} results in:
\begin{align}
 - \bar{T}_{a,(i,j)} \leq \Delta T_{a,(i,j)}^{\infty} + T^\mathrm{DA}_{a,(i,j)} \leq \bar{T}_{a,(i,j)}
\end{align}
and identical relation also holds for $a'$.

\vspace{-6pt}
\section{The Proof of Theorem \ref{theorem::KKtEquiv}}
\label{App::KKtEquiv}

We prove Theorem \ref{theorem::KKtEquiv} for the special case of two regions with one tie-line in between them, and the proof can be extended to any number of regions and tie-lines. We denote these two regions with $a$ and $a'$.
Now we are going to compare the set of KKT optimality conditions for the centralized OPF in \eqref{eq::cen_obj_fun}-\eqref{eq::cen_theta_1}, denote this set with $\mc{S}_c$, versus the aggregates set of KKT conditions of the decentralized OPF in \eqref{eq::decen_obj_fun}-\eqref{eq::decen_theta_1} for regions $a$ and $a'$ (we denote this set of aggregated KKT optimality equations with $S_d$). Now one can easily verify that the set of equations in $S_c$ and $S_d$ are the same with two differences:
\begin{itemize}
    \item[$1$:] The equations on $S_c$ has three extra equation corresponding to the primal feasibility of \eqref{eq::TcenMarkIneq}. We know the final values derived from the proposed iterative mechanism satisfy \eqref{eq::TcenMarkIneq} (Lemma \ref{lemma::Xfeasible}). Hence, we can add the corresponding primal feasibility conditions of \eqref{eq::TcenMarkIneq} to $S_d$.
    \item[$2$:] After updating $S_d$ in step 1, the only difference between $S_c$ and $S_d$ is because of the derivatives of the Lagrangian function w.r.t. the flow on the tie-line. Let's assume that the direction of the flow of the tie-line is from region $a$ to region $a'$. In the $S_c$ we have one equation for this:
    \begin{align}
        & \frac{\partial \mc{L}_0}{\partial \Delta T_{a,(i,j)}} = \label{eq::DervLagC}  \\
        & \alpha_{a,i}^{\ast} - \alpha_{a',j}^{\ast} + \xi_{a,(i,j)}^{\ast} - \bar{\kappa}_{a,(i,j)}^{\ast} + \bar{\eta}_{a,(i,j)}^{\ast} 
    + \gamma_a^{\ast} - \gamma_{a'}^{\ast} = 0 \nn
    \end{align}
    and because we have two variables in the decentralized OPF for the flow on this tie-line, {i.e.}, $T_{a,(i,j)}$ and $T_{a',(j,i)}$, we have two equations corresponding to the derivative of the Lagrangian function w.r.t. tie-line flow: for area $a$,    
    \begin{align}
        & \frac{\partial \mc{L}_a}{\partial \Delta T_{a,(i,j)}} = \hat{\alpha}_{a,i}^{\infty} + \hat{\xi}_{a,(i,j)}^{\infty} + \hat{\gamma}_a^{\infty} - \delta_{a',(i,j)}^{\infty} \nn \\
        & \hspace{47pt} + \frac{\mu_{(i,j)}^{\infty}}{2} \mbox{sign} (\Delta T_{a, (i,j)}^{\infty}) \nn \\
        & \hspace{29pt} = \hat{\alpha}_{a,i}^{\infty} + \hat{\xi}_{a,(i,j)}^{\infty} + \hat{\gamma}_a^{\infty}
        - \hat{\alpha}_{a',j}^{\infty} \nn \\
        & \hspace{47pt} - \hat{\gamma}_{a'}^{\infty} + \frac{\mu_{(i,j)}^{\infty}}{2} \mbox{sign} (\Delta T_{a, (i,j)}^{\infty}) = 0   \label{eq::DervLaga}
    \end{align}
    and similarly for area $a'$,
    \begin{align}
        & \frac{\partial \mc{L}_{a'}}{\partial \Delta T_{a',(j,i)}} 
         = \hat{\alpha}_{a',j}^{\infty} + \hat{\xi}_{a',(j,i)}^{\infty} + \hat{\gamma}_{a'}^{\infty} - \hat{\alpha}_{a,i}^{\infty} \nn \\
        & \hspace{47pt} - \hat{\gamma}_{a}^{\infty} - \frac{\mu_{(i,j)}^{\infty}}{2} \mbox{sign} (\Delta T_{a,(i,j)}^{\infty}) = 0  \label{eq::DervLagaprime}
    \end{align}
\end{itemize}
where the equality in \eqref{eq::DervLaga} and \eqref{eq::DervLagaprime} come from replacing $\delta_{a, (j,i)}^{\infty}$ and $\delta_{a',(i,j)}^{\infty}$ with their final values (cf. Section \ref{5_Converg}, Item 2).
Note that from \eqref{eq::DervLaga} and \eqref{eq::DervLagaprime} we have $\hat{\xi}_{a',(i,j)}^{\infty} + \hat{\xi}_{a,(j,i)}^{\infty} = 0$, which means that $\hat{\xi}_{a',(i,j)}^{\infty} = - \hat{\xi}_{a,(j,i)}^{\infty}$. Hence, $\frac{\partial \mc{L}_{a'}}{\partial \Delta T_{a',(j,i)}} = -\frac{\partial \mc{L}_a}{\partial \Delta T_{a,(i,j)}}$. This means that the set of equations in $\mc{S}_c$ and $\mc{S}_d$ are equivalent except that $\mc{S}_c$ has equation \eqref{eq::DervLagC} and $\mc{S}_d$ has equation \eqref{eq::DervLaga}. 
Now we show that from any solution to the set of equations in $\mc{S}_d$, which is the KKT conditions of the decentralized OPF in \eqref{eq::decen_obj_fun}-\eqref{eq::decen_theta_1}, we can build the solution to the set of equations in $\mc{S}_c$, which is the KKT conditions of the centralized OPF in \eqref{eq::cen_obj_fun}-\eqref{eq::cen_theta_1}. The corresponding values of the variables $\bar{\kappa}_{a,(i,j)}^{\star}$ and $\bar{\eta}_{a,(i,j)}^{\star}$ for the solution of $\mc{S}_c$ would be:
\begin{align}
    \bar{\kappa}_{a,(i,j)}^{\star} &\triangleq \min( \frac{\mu_{(i,j)}^{\infty}}{2} \mbox{sign} (\Delta T_{a,(i,j)}^{\infty}), 0) \nn \\
    \bar{\eta}_{a,(i,j)}^{\star} &\triangleq \max( \frac{\mu_{(i,j)}^{\infty}}{2} \mbox{sign} (\Delta T_{a,(i,j)}^{\infty}), 0)
\end{align}
and the value of the remaining variables in the $\mc{S}_c$ is equal to the values of the corresponding variables at the solution for the equations in $\mc{S}_d$. Similar arguments can be made for the case with arbitrary number of regions and tie-lines.

 
\ifCLASSOPTIONcaptionsoff
https://www.overleaf.com/project/60341c595786c7a81a9cb990  \newpage
\fi

\bibliographystyle{IEEEtran}
{\bibliography{main}}

\end{document}